\documentclass[conference, compsocconf]{IEEEtran}
\usepackage{amssymb}
\usepackage{amsmath}
\usepackage{graphicx}
\usepackage{gastex}
\usepackage{amsthm}
\usepackage{thm-restate}
\usepackage{booktabs}
\usepackage{multirow}
\usepackage[table]{xcolor}
\usepackage{colortbl}
\usepackage{algorithm}
\usepackage[noend]{algpseudocode}
\usepackage{fancyvrb}
\usepackage{multirow}
\usepackage{wrapfig}
\usepackage{tikz}
\usepackage{url}
\usetikzlibrary{calc,trees,arrows,automata,shapes,snakes, positioning}
\usepackage{pgfplots}
\usepgfplotslibrary{patchplots}
\pgfplotsset{compat=1.3}
\usepackage{array}
\usepackage{diagbox}
\usepackage{amsthm}
\usepackage{mathptmx}

\usetikzlibrary{calc,trees,shapes,intersections,arrows,automata,patterns,plotmarks}
\tikzset{astate/.style={draw,circle,inner sep=0pt,minimum size=8mm}}
\tikzset{astateq/.style={draw,rectangle,inner sep=3pt,rounded corners}}
\tikzset{adistr/.style={draw,circle,fill,minimum size=1mm,inner sep=0mm}}

\colorlet{colortime}{red}
\colorlet{colormem}{blue}

\pgfplotsset{
	/pgfplots/ybar legend/.style={
		/pgfplots/legend image code/.code={%
		\draw[##1,/tikz/.cd,bar width=3pt,yshift=-0.3em,bar shift=0pt,draw=none]
			plot coordinates {(0mm,0.8em)};},
	},
	every axis y label/.style={at={(ticklabel cs:0.5)}, rotate=90,anchor=near ticklabel},
	every axis x label/.style={at={(ticklabel cs:0.5)}, anchor=near ticklabel},
	every axis legend/.append style={draw=none,column sep=4mm},
	every axis/.append style={clip=false}
}
\newtheorem{definition}{Definition}
\newtheorem{theorem}{Theorem}
\newtheorem{corollary}{Corollary}
\newtheorem{example}{Example}
\newtheorem{lemma}{Lemma}

\renewcommand{\phi}{\varphi}
\newcommand{\len}[1]{|#1|}
\newcommand{\states}{S}
\newcommand{\concat}{\cdot}

\newcommand{\tranRel}{\text {\sf T}}
\newcommand{\trace}{\mathit{trace}}
\newcommand{\Kripke}{\mathcal{K}}
\newcommand{\aKripke}{(\states,\initialState,\tranRel,\AP,\labels)}
\newcommand{\labels}{L}
\newcommand{\AP}{\mathit{AP}}
\newcommand{\prefix}[2]{#1|^{#2}}
\newcommand{\suffix}[2]{#1|_{#2}}

\newcommand{\LTL}{\text{LTL}}
\newcommand{\UXLTL}{\text {\sf LTL(U, X)}}
\newcommand{\U}{\text{\sf U}}
\newcommand{\X}{\text{\sf X}}

\newcommand{\W}{\text{\sf W}}

\newcommand{\initialState}{\bar{s}}
\newcommand{\paths}{\pi}
\newcommand{\pathSet}{\mathit{Paths}}
\newcommand{\pathAt}[2]{#1[#2]}

\newcommand{\F}{\text{\sf F}}
\newcommand{\G}{\text{\sf G}}
\renewcommand{\bot}{\text{\bf 0}}
\renewcommand{\top}{\text{\bf 1}}
\newcommand{\FG}{\text{\sf LTL(F, G)}}

\newcommand{\size}[1]{|#1|}
\renewcommand{\flat}{\mathit{fnf}}
\newcolumntype{M}[1]{>{\centering\arraybackslash}m{\dimexpr#1\linewidth-2\tabcolsep}}

\def\implies{\to}

\begin{document}

\title{Verify LTL with Fairness Assumptions Efficiently}

\author{
\IEEEauthorblockN{Yong Li\IEEEauthorrefmark{1}\IEEEauthorrefmark{2}, Lei Song\IEEEauthorrefmark{1}, Yuan Feng\IEEEauthorrefmark{3}, Lijun Zhang\IEEEauthorrefmark{1}\IEEEauthorrefmark{2}}
\IEEEauthorblockA{\IEEEauthorrefmark{1}State Key Laboratory of Computer Science, Institute of Software,
 CAS, China}
\IEEEauthorblockA{\IEEEauthorrefmark{2}University of Chinese Academy of Sciences, China}
\IEEEauthorblockA{\IEEEauthorrefmark{3}Centre for Quantum Computation and Intelligent Systems,\\
University of Technology Sydney, Australia}
}

\maketitle

\begin{abstract}
This paper deals with model checking problems with respect to LTL
properties under fairness assumptions. We first present an
efficient algorithm to deal with a fragment of fairness
assumptions and then extend the algorithm to handle arbitrary
ones. Notably, by making use of some syntactic
transformations, our algorithm avoids constructing
corresponding B\"uchi automata for the whole fairness assumptions,
which can be very large in practice. We implement our algorithm in
NuSMV and consider a large selection of formulas. Our experiments show
that in many cases our approach exceeds the automata-theoretic
approach up to several orders of magnitude, in both time and memory.
\end{abstract}


\IEEEpeerreviewmaketitle

\section{Introduction}
\label{sec:intro}
\emph{Linear Temporal Logic} (LTL)~\cite{Pnueli1977TLP} has been shown to be a proper
specification language. As a result, for verifying reactive systems, model checkers
for $\LTL$, like Spin~\cite{Holzmann1997MCS} and
NuSMV~\cite{DBLP:journals/sttt/CimattiCGR00}, have been applied in
practice successfully. To verify
whether or not a system satisfies an $\LTL$ formula, the
classical automata-theoretic approach~\cite{DBLP:conf/lics/VardiW86} is
usually adopted: Firstly, a B\"uchi automaton is built which accepts
all executions violating the $\LTL$ formula; Secondly, a product
system is built from the original system and the B\"uchi automaton;
Finally, the problem is reduced to finding an accepting path in the
product system. Since in the worst case the constructed B\"uchi automaton
can be exponentially larger than the $\LTL$ formula, both time and
space complexity of the algorithm in~\cite{DBLP:conf/lics/VardiW86} is exponential with
respect to the size of the $\LTL$
formula. The complexity of this algorithm is shown to
be PSPACE-complete
in~\cite{Sistla:1985:CPL:3828.3837}. Even if we restrict to a small subset of $\LTL$
formulas (those only containing eventual modality $\F$), it is still
NP-complete. On the other side, due to the popularity of $\LTL$, many ideas have
been proposed optimizing the construction of B\"uchi automata, see
e.g.~\cite{DBLP:conf/cav/DanieleGV99, DBLP:conf/concur/EtessamiH00,
Gastin2001, DBLP:conf/cav/SomenziB00, Kini2015tacas, Sickert2016cav}.

The classification of properties into different categories is pivotal
for efficient verification of reactive systems. In the seminal paper~\cite{DBLP:dblp_journals/tse/Lamport77},
Lamport introduced the notions of safety and liveness properties,
where ``safety'' properties assert something ``bad'' never happens,
while ``liveness'' properties require something ``good'' will
eventually happen. These notions were later formalized by Alpern and
Schneider in~\cite{alpern1987recognizing}.
Properties were further classified into
strong safety and absolute liveness
in~\cite{DBLP:journals/fac/Sistla94}, and fair properties. The notion of fairness
is important for verifying liveness
in reactive systems to remove unrealistic
behaviors~\cite{Francez:1986:FAI:19247}.

In practice, fairness assumptions can have a great impact on the performance
in many cases. For instance in the
binary semaphore protocol~\cite{HammerMAM05}, the fairness assumption that whenever a process is ready,
it will
have a chance to enter the critical section, is given by: $\bigwedge_{1\le
  i\le n}(\G\F\mathit{enter}_i\implies\G\F\mathit{critical}_i)$ ($\G$
and $\F$ denote ``always'' and ``eventually'', respectively) with
$n$ being the number of processes. When $n=5$, the corresponding B\"uchi automaton
generated by LTL3BA~\cite{Gastin2001} has more than 300
states and 1 million transitions\footnote{Interested readers can try
  the online $\LTL$ translator available at: \url{http://spot.lip6.fr/ltl2tgba.html}}.
Therefore, model checking formulas under such an assumption
will be time and memory consuming even when given formulas are simple.

In this paper we propose a novel algorithm to verify fairness as well as general
properties with fairness assumptions.
We do not only consider simple fairness formulas as mentioned above, but also
consider more complex fairness with nested modalities like $\F\G(a\lor\F b)$. Moreover,
 we extend the notion of fairness assumptions to full $\LTL$ formulas, which allows us
 to specify some fairness assumption like ``$a$ and $ \X(b\U c)$ holds infinitely often''.
Notably, our algorithm relies on
a syntactic transformation and avoids constructing a B\"uchi
automaton for the whole fairness. The approach is presented in three steps:
\begin{itemize}
\item We first restrict to fairness with only $\F$ and $\G$
  modalities, for which our syntactic transformation can completely avoid
  B\"uchi automata construction. For this setting our approach
  achieves a speedup of \emph{four orders} of magnitudes on some examples.
\item We then extend the algorithm to deal with fair formulas of full $\LTL$. The idea is to
  transform a fair formula into an equivalent one in disjunctive norm
  form, each sub-formula of which can be handled by specific and
  efficient algorithms. Even though we may still resort to the
  automata-theoretic approach for some sub-formulas, they are often
  much smaller than the original one.
\item Finally, we show our approach can be adapted to accelerate the
  verification of generic $\LTL$ formulas under fairness assumptions.
\end{itemize}

We have implemented the algorithm in NuSMV and compared it with the
classical algorithm. Our experimental results show that for many cases
while NuSMV runs out of time or memory quickly, our algorithm
terminates within seconds using memory less than 100 MB. The
main reason is that after the syntactical transformation, we can avoid
constructing B\"uchi automata for many sub-formulas. Even for those
where B\"uchi automata construction is inevitable, their corresponding
automata are relatively small and can be constructed efficiently.

It should be pointed out, however, that the syntactical transformation may
also cause exponential blow-ups in the length of given formulas.
Hence, as the experimental results show, our algorithm may be much slower
than NuSMV in some cases.
We then further discuss and characterize the formulas
for which our approach provides better performance.

\paragraph{Related Work}\label{sec:related}
There is a plenty of work on optimizing verification of $\LTL$ (or its
sub-logic). Here we only briefly recall a few closely
related works. In~\cite{BloemRS99}, specialized algorithms are proposed to deal with
$\LTL$ properties, which can be represented by either \emph{terminal} or \emph{weak} automata.
Compared to general algorithms, specialized algorithms improve the worst-case
time complexity by a constant factor.
This result is further formalized in~\cite{CernaP03}, where it is shown that
terminal and weak automata correspond to \emph{guarantee} properties
(something happens eventually) and \emph{persistence}  properties (something
always happens eventually), respectively. For guarantee properties, model checking algorithm
reduces to the reachability of an accepting state, while for persistence properties,
it reduces to finding a fully accepting cycle, i.e., all states on it are accepting.
Furthermore, a decision algorithm is proposed in~\cite{CernaP03} to check
whether an $\LTL$ formula is a guarantee or persistence property.
For properties which are neither guaranteed nor persistent, the general algorithm has to be
used. One exception is~\cite{RenaultDKP13}, where a decomposition algorithm is proposed
for strong automata, which are neither terminal nor weak.
The idea is to decompose a strong automaton into three sub-automata,
which are terminal, weak, and strong, respectively.
Then specialized algorithms can be used to check the terminal and weak sub-automata.
Since the strong sub-automaton is smaller than the original automaton,
decomposition always speeds up the verification according to the experiment in~\cite{RenaultDKP13}.

Differently,
our algorithm performs decomposition \emph{syntactically} on given
  formulas, hence we do not need to build their corresponding B\"uchi
  automata at the beginning. While the specialized algorithm
  in~\cite{BloemRS99,CernaP03,RenaultDKP13} is automata-based, B\"uchi
  automata have to be built beforehand, which may take a significant
  amount of time and memory, especially when the given formula is
  long~\cite{HammerMAM05}. Moreover, our algorithm works for arbitrary
  fairness including those which are neither guaranteed nor
  persistent.

\paragraph{Organization of the paper.}
\label{sec:orignization-paper}
Section~\ref{sec:pre} introduces some definitions and notations used
throughout the paper. The algorithm is presented in Section~\ref{sec:algo}.
We demonstrate the efficiency
of our algorithm via experiment in
Section~\ref{sec:experiment}. Finally, we conclude our paper
in Section~\ref{sec:concl-future-work}.

All missing proofs can be found in the appendix.

\section{Preliminaries}
\label{sec:pre}

We shall first introduce some preliminary definitions and notations and then
present the syntax and semantics of $\LTL$.

Let $X$ be a finite set of elements and $\xi=x_0x_1\ldots\in X^*$
with $X^*=\cup_{i\ge 0}X^i$ a finite sequence of elements in $X$. For
each $\xi\in X^i$, we let $\len{\xi}=i+1$ denote its length.
An infinite sequence $\xi\in X^\omega$ is \emph{cyclic} if there
exists $\xi'\in X^i$ for some $i$ such that
$\xi=(\xi')^\omega$, i.e., repeating $\xi'$ for infinite times.
Let $\pathAt{\xi}{i}=x_i$ denote the $(i+1)$-th element on $\xi$ if it exists.
We shall write $\prefix{\xi}{i}$ to denote the prefix of $\xi$
ending at the $(i+1)$-th element, while $\suffix{\xi}{i}$ the suffix of $\xi$
starting from the $(i+1)$-th element. Let $\xi\in X^*$ and $\xi'\in X^\omega$. Then
$\xi\concat\xi'$ denotes the infinite sequence obtained by
attaching $\xi'$ to the end of $\xi$.

We will fix a finite set of atomic propositions, denoted $\AP$ and
ranged over by $a,b,c,\ldots$, throughout the remainder of the
paper. The syntax of $\LTL$ is given by the following grammar:
$$
\phi,\psi ::= a \mid \neg a \mid \phi_1\land\phi_2 \mid \phi_1\lor\phi_2 \mid \X\phi \mid \phi_1\U\phi_2 \mid \phi_1\W\phi_2
$$
where $a\in\AP$ and $\phi,\psi,\phi_1,$ and $\phi_2$ range over $\LTL$ formulas.
As usual, we introduce some abbreviations: $\top=a\lor\neg a$ and
$\bot=a\land\neg a$ denote $\mathit{True}$ and
$\mathit{False}$ respectively,
while $\F\phi = \top\U\phi$ (eventually $\phi$), $\G\phi=\phi\W\bot$
(always $\phi$), and
$(\phi_1\implies\phi_2)=(\neg\phi_1\lor\phi_2)$. Let $l,l_1,l_2,\ldots$
range over \textit{propositional formulas}, i.e.,
formulas defined by: $l::= a\mid\neg a\mid l_1\land l_2 \mid l_1\lor l_2.$

Given an infinite sequence of sets of atomic propositions
$\rho=A_0A_1\ldots\in (2^\AP)^\omega$ and an $\LTL$ formula $\phi$, we say
$\rho$ satisfies $\phi$, written as $\rho\models\phi$, if:
$$
\begin{array}{rcl}
  \rho\models a &\text{ iff }&a\in\pathAt{\rho}{0}\\
  \rho\models \X\phi &\text{ iff }&\suffix{\rho}{1}\models\phi\\
  \rho\models \phi_1\U\phi_2 &\text{ iff }&\exists i\ge 0.(\suffix{\rho}{i}\models\phi_2\land\forall 0\le j<i.\suffix{\rho}{j}\models\phi_1)\\
  \rho\models \phi_1\W\phi_2 &\text{ iff }&(\forall i\ge 0.\suffix{\rho}{i}\models\phi_1)
\lor(\rho\models\phi_1\U\phi_2)
\end{array}
$$
All other connectives are defined in a standard way.
For formulas $\phi$ and $\psi$, we say that $\phi$ and $\psi$ are semantically
equivalent, denoted $\phi\equiv \psi$, if $\rho\models\phi$ iff
$\rho\models\psi$ for any $\rho\in(2^\AP)^\omega$.

Here we only define $\LTL$ formulas in \textit{positive normal form}, in the sense that
the negation operator can only be applied to atomic propositions. However, it is well-known that
any $\LTL$ formula can be transformed into an equivalent one
in positive normal form, using $\neg(\X\psi) \equiv \X(\neg\psi)$ and the following \emph{duality laws}:
$$
\begin{array}{ccc}
  \neg(\phi_1\U\phi_2) &\equiv&
  (\phi_1\land\neg\phi_2)\W(\neg\phi_1\land\neg\phi_2)\\
  \neg(\phi_1\W\phi_2) &\equiv& (\phi_1\land\neg\phi_2)\U(\neg\phi_1\land\neg\phi_2)
\end{array}
$$

Fairness assumptions are critical to rule out unrealistic behaviors
when performing verification; see for
instance~\cite{DBLP:journals/acta/QueilleS83,Francez:1986:FAI:19247}. Formally,
fairness is a fragment of $\LTL$, which can be defined as follows:

\begin{definition}[{\cite{DBLP:journals/fac/Sistla94}}]\label{def:fair}
  An $\LTL$ formula $\phi$ is a \emph{fairness} iff for any $\rho\in(2^\AP)^\omega$,
  \begin{enumerate}
  \item the set of sequences satisfying $\phi$ is
  closed under suffixes, i.e., $\rho\models\phi$ implies
  $\suffix{\rho}{i}\models\phi$ for any $i\ge 0$;
 \item the set of sequences satisfying $\phi$ is closed under
   prefixes, i.e., $\rho\models\phi$ implies
   $\rho_1\concat\rho\models\phi$ for any
   $\rho_1\in(2^\AP)^*$.
  \end{enumerate}
\end{definition}

We shall refer properties defined in
Definition~\ref{def:fair} as \emph{fair formulas} or \emph{fairness} in the following.
According to Definition~\ref{def:fair}, the following lemma is straightforward:
\begin{lemma}\label{lemma:fair-chara}
  $\phi$ is a fairness iff $\phi\equiv\G\phi$ and $\phi\equiv\F\phi$.
\end{lemma}
As a result of Lemma~\ref{lemma:fair-chara}, we can add any number of $\F$
and $\G$ in front of a fairness without changing its semantics. For instance,
fairness $\F a \lor \G\neg a$ is equivalent to $\G\F(\F a \lor \G\neg a)$.

As usual we consider models given as Kripke structures, which are
formally defined as follows:
\begin{definition}\label{def:kripke}
  A \emph{Kripke structure} is a tuple $\Kripke:=\aKripke$
  where $\states$ is a finite set of states, $\initialState\in\states$ is the initial state,
  $\tranRel\subseteq\states\times\states$ is a set of transitions, and
  $\labels:\states\rightarrow 2^\AP$ is a labeling
  function. We
  assume that for each $s\in\states$, there exists $s'\in\states$
  such that $(s,s')\in\tranRel$.
\end{definition}

We fix a Kripke structure $\Kripke=\aKripke$
throughout the remainder of the paper. Let $r,s,t,\ldots$ range over
$\states$. Let $\pathSet^\omega(s)\subseteq\states^\omega$ denote the set
of infinite paths starting from $s$ such that
$\paths\in\pathSet^\omega(s)$ iff $\pathAt{\paths}{0}=s$ and for any $i\ge 0$,
$(\pathAt{\paths}{i},\pathAt{\paths}{i+1})\in\tranRel$. Similarly, we can define $\pathSet^*(s)$, i.e., finite paths
in $\Kripke$ starting from $s$. Let
$\pathSet^\omega(\Kripke)=\pathSet^\omega(\initialState)$ and
$\pathSet^*(\Kripke)=\pathSet^*(\initialState)$.
Given $\paths\in\states^\omega$, let $\trace(\paths)$ denote the trace of $\paths$ such that
$\pathAt{\trace(\paths)}{i}=\labels(\pathAt{\paths}{i})$ for all $i\ge
0$, i.e., $\trace(\paths)$ denotes the sequence of labels of states in
$\paths$. For an $\LTL$ formula $\phi$, we write $\paths\models\phi$ iff
$\trace(\paths)\models\phi$; $s\models\phi$ iff $\paths\models\phi$ for all
$\paths\in\pathSet^\omega(s)$; $\Kripke\models\phi$ iff
$\initialState\models\phi$. Given an $\LTL$ formula $\phi$ and a
fairness $\phi_f$, $\Kripke$ satisfies $\phi$ under the assumption
$\phi_f$ iff $\Kripke\models(\phi_f\implies\phi)$.
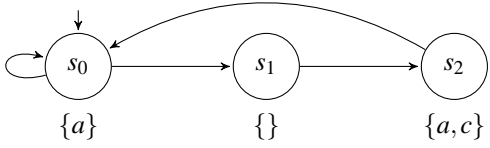
\begin{figure}
	\centering
	\begin{tikzpicture}[>=stealth',shorten >=1pt,auto]	
    \node (K) at (0,0) {};

    \node[state] (s0) at ($(K) + (0,0)$) {$s_0$};
	\node[state] (s1) at ($(K) + (2.5,0)$) {$s_1$};
	\node[state] (s2) at ($(K) + (5,0)$) {$s_2$};

    \node (s0L) at ($(s0) + (0, -0.8)$) {$\{a\}$};
    \node (s1L) at ($(s1) + (0, -0.8)$) {$\{\}$};
    \node (s2L) at ($(s2) + (0, -0.8)$) {$\{a,c\}$};

	\draw[->] ($(s0.north) + (0,0.3)$) to (s0.north);
	\draw[->] (s0) edge [loop left] node {} (s0);
	\draw[->] (s0) edge  node {} (s1);
	\draw[->] (s1) edge node {} (s2);
	\draw[->] (s2) edge [bend right] node {} (s0);
	\end{tikzpicture}
	\caption{An example of Kripke structure}
	\label{fig:exampleKripke}
\end{figure}

\begin{example}\label{ex:kripke}
An example for Kripke structure is $\Kripke = (\{s_0, s_1, s_2\},s_0,\tranRel,\{a,b,c\},\labels)$,
where $\tranRel$ and $\labels$ are depicted in Figure~\ref{fig:exampleKripke}, for instance $\labels(s_0) = \{a\}$.
Obviously, traces in $\Kripke$ can be represented as $(\{a\}^*\{\}\{a,c\})^{*}\{a\}^\omega \mid (\{a\}^*\{\}\{a,c\})^\omega$.

\end{example}
Moreover, we directly conclude the corollary below from Lemma~\ref{lemma:fair-chara}:
\begin{corollary}\label{coro:suffixpath}
Let $\paths\in\pathSet^\omega(\Kripke)$ and $\phi$ a fairness. Then for any index $j\geq 0$, $\paths\models \phi$ iff $\suffix{\paths}{j}\models \phi$.
\end{corollary}
\begin{proof}
  Since $\phi$ is a fairness, $\phi \equiv \G\phi$. Thus $\paths\models\phi$
  implies $\suffix{\paths}{j}\models \phi$ for any index $j$. On the other hand,
  $\suffix{\paths}{j}\models \phi$ implies $\paths\models \F\phi$. Then
  we conclude $\paths\models \phi$ by $\F\phi\equiv\phi$.
\end{proof}

Intuitively, we can safely consider only the suffixes of the paths when the given formula
is a fairness.
\section{Model Checking Fairness}
\label{sec:algo}
In this section, we present
an algorithm for model checking fair formulas. We first describe the overall idea.
For fair formula $\phi$, $\Kripke\models\phi$ means that for all infinite paths $\paths$ starting from initial state $\initialState$, $\paths\models\phi$.
Conversely, if $\neg(\Kripke\models\phi)$, then there exists an infinite path $\paths$ such that
 $\paths\models\neg\phi$. Thus, we first construct the negation $\neg\phi$, which is also a fair formula. Then, we
 construct a \emph{fair normal form} of $\neg\phi$, denoted by $\flat(\neg\phi)$, which has the form
 $\lor_{i=1}^m \phi_i$, with each $\phi_i$ being a fair formula. Then, the problem reduces to checking whether
there exists an infinite path $\paths$ such that $\paths\models\phi_i$. In other words, whether there
exists an SCC $B$ \emph{satisfying} for $\phi_i$: the satisfaction here can be checked by analysing the SCC $B$.
A strongly connected component (SCC) $B$ of $\Kripke$ is a state set such that for any $s,t\in B$, there exists a path from $s$ to $t$.
We here do not consider trivial SCCs which are single states without self loops.

We start with treating fairness formulas in $\FG$, then we
extend the algorithm to deal with general fairness. Finally, we handle all $\LTL$ formulas with fairness assumptions.
\subsection{Fairness in $\FG$}
\label{sec:fgfg}
In this subsection we focus on a fragment of $\LTL$ formulas, denoted
$\FG$, which only contains $\F$ and $\G$ modalities, i.e., it is defined
by the following grammar:
$$
  \phi ::= a \mid \neg a \mid \phi\land\phi \mid \phi\lor\phi \mid \F\phi \mid \G\phi.
$$

For each fairness in $\FG$, we shall show that it can be
transformed into an equivalent formula where all
propositional formulas are directly preceded by precisely two modalities, either
$\F\G$ or $\G\F$.
Such a transformation is purely syntactical: we call the transformation
procedure the \emph{flatten} operation, denoted by $\flat$.

\begin{theorem}\label{thm:normform-fg}
Let $\phi\in\FG$ be a fairness. Then, it can be transformed into the following equivalent formula, referred to also as
its fair normal form:
$$\flat(\phi) := \bigvee_{i=1}^m \left(\F\G l_i\land(\bigwedge_{j=1}^{n_i} \G\F l_{ij})\right)$$
where $l_i$ and $l_{i,j}$ are propositional formulas.
\end{theorem}
Note that $m$ and $n_i$ are nonnegative integers and we omit $\F\G l_i$ and $\G\F l_{ij}$ in the fair normal form
whenever $l_i$ and $l_{ij}$ are $\top$.
The syntactic transformation $\flat$ is the key of the algorithm: $\flat(\phi)$ can be checked directly on $\Kripke$ without constructing the product automaton.
We first give an example to illustrate the main stapes of verifying fair formulas in $\FG$.

\begin{example}\label{ex:flatncheck}
Take $\phi=\neg(\F\G(a \lor (\F b \land \G c)))$ for example, the fair normal form of $\neg\phi$
is $\flat(\neg\phi)=\F\G a \lor(\F\G c \land \G\F b)$. Consider model checking the Kripke structure
$\Kripke$ in Example~\ref{ex:kripke} against fair formula $\phi$. We already have the
fair normal form of $\neg\phi$ above, so we only need to check whether there exists an SCC satisfying
fair formula $\F\G a$ or $\F\G c\land \G\F b$. Consider fair formula $\F\G a$, we find that there exists an SCC
$\{s_0\}$ reachable from initial state $s_0$ that fulfils the formula. We therefore conclude that $\Kripke$ does not
satisfy $\phi$ and give a counterexample $\paths=(s_0)^\omega$ such that $\paths\models\neg\phi$, thus $\paths\not\models\phi$.
\end{example}

We below give the intuition behind the syntactic transformation.

First, we have to deal with  trivial fair formula such as $\F a\lor \G \neg a$.
Due to Lemma~\ref{lemma:fair-chara}, we first add $\G\F$ in front of the original formula
and then apply the flatten operation, which gives us the fair normal form $\G\F a \lor \F\G \neg a$.

Given a fairness $\phi\in\FG$, our goal is to obtain an equivalent formula of the fair normal form.
To that end, we first make sure that there exists at least one $\F\G$ or $\G\F$ in front of every
propositional formula, which is guaranteed by safely adding $\G\F$ in front of $\phi$. After that, we
are going to push every $\F\G$ and $\G\F$ directly in front of all propositional formulas. To achieve
this, one needs to discuss the distributivity of $\G\F$ and $\F\G$ over $\lor$ and $\land$.

Suppose $\phi_1, \phi_2\in \LTL$, our goal is pushing $\G\F$ and $\F\G$ inside
such that they appear only before propositional formulas.
We consider the following four cases:
\begin{enumerate}
\item\label{rm-dup} $\G\F\G\equiv\F\G, \G\F\F\equiv\G\F, \F\G\G\equiv\F\G$ and $\F\G\F\equiv\G\F$ are trivial
according to the semantics of $\LTL$. This insures that we have only $\G\F$ and $\F\G$
modalities since we first add $\G\F$ in front of $\phi$.
\item\label{push-dis} $\G\F(\phi_1\lor\phi_2)\equiv \G\F\phi_1\lor\G\F\phi_2$
and $\F\G(\phi_1\land\phi_2)\equiv\F\G\phi_1\land\F\G\phi_2$ hold since $\G\F$ and $\F\G$ are
distributive over $\lor$ and $\land$ operator respectively.
\item\label{push-diff} $\G\F(\phi_1\land\F \phi_2)\equiv\G\F\phi_1\land\G\F\phi_2$,
$\F\G(\phi_1\lor\G\phi_2)\equiv\F\G\phi_1\lor\F\G\phi_2$,
$\G\F(\phi_1\land\G\phi_2)\equiv\G\F\phi_1\land\F\G\phi_2$ and
$\F\G(\phi_1\lor\F\phi_2)\equiv\F\G\phi_1\lor\G\F\phi_2$. Intuitively, if the operands of
$\G\F$ and $\F\G$ are not propositional formulas, they must be the four cases we listed here
after we go through case~\ref{push-dis}) and case~\ref{push-ant-dis}).
\item\label{push-ant-dis} $\G\F (\phi_1\land \phi_2)$ and $\F\G (\phi_1\lor\phi_2)$. This is the
most challenging part since $\G\F$ ($\F\G$) is not distributive over $\land$ ($\lor$) operator.
The following procedure relies on the structure of the formula. If the operand of $\G\F$ or $\F\G$ is propositional formula, then it is already the formula we desire.
Otherwise, if they are not case~\ref{push-diff}) such as the formula $\F\G(a \lor (\F b \land \G c))$, we transform $\phi_1\land\phi_2$ and $\phi_1\lor\phi_2$ to disjunctive normal form (DNF) and conjunctive normal form (CNF) respectively. After that, we apply case~\ref{push-dis}) and may use case~\ref{push-diff}) for further processing.

\end{enumerate}
Once we get a formula where all propositional formulas are adjacent to $\G\F$ or $\F\G$, we
transform it into DNF, which gives us the fair normal form in Theorem~\ref{thm:normform-fg}.
We illustrate the procedure of $\flat$ operator via an
example as follows:
\begin{example}\label{ex:flat}
Let $\phi=\F\G(a \lor (\F b \land \G c))$. We show how to flatten $\phi$ step by step.
\begin{itemize}
\item Add $\G\F$ in front of $\phi$ which gives us $\phi$ since $\G\F\F\G\equiv\F\G$;
\item Since $\phi$ is an instance of case \ref{push-ant-dis}), we first transform $a \lor (\F b \land \G c)$
into a CNF, which results in $\F\G((a \lor \F b) \land (a\lor \G c))$;
\item According to case \ref{push-dis}), $\F\G$ is distributive over $\land$ operator, we therefore directly
push $\F\G$ inside, which gives us $\F\G(a \lor \F b) \land \F\G (a\lor \G c)$;
\item The resulting formula is an instance of case $\ref{push-diff}$), we get $ (\F\G a \lor \G\F b)\land (\F\G a \lor \F\G c)$
after we apply the equations in case $\ref{push-diff}$).
\item Since all $\F\G$ are adjacent to propositional formulas, we transform above formula to DNF, which gives us
the fair normal form $\F\G a \lor (\F\G c\land\G\F b)$ of $\phi$.
\end{itemize}
Intuitively, it means whenever $\paths\models\phi$, it must be the
case that $\paths$ ends up with a loop such that either all states on
the loop satisfy $a$ or all states satisfy $c$ and at least one state satisfies $b$.
This also can be verified by applying the semantics of $\LTL$.

\end{example}
By Corollary~\ref{coro:suffixpath}, we only need to
consider the infinite suffixes of the paths that all states will be visited infinitely often. That is to say, we only need to consider all the SCCs of $\Kripke$ that can be
reached.

\begin{definition}[Accepting SCC]\label{def:accepting-SCC-FG}
Given a formula $\phi = \F\G l \land (\bigwedge_{j=1}^{m} \G\F
 l_{j})$ and an SCC $B$. If 1) for every state $s \in B$, $s\models l$ and 2) for each $j$, there exists $s\in B$, such that $s \models l_j$, then we say SCC $B$ is accepting for $\phi$.
\end{definition}
With the definition of accepting SCC, we have the following theorem:
\begin{theorem}\label{thm:existspath-FG}
 For any $\phi = \F\G l \land (\bigwedge_{j=1}^{m} \G\F
 l_{j}) $, there exists an infinite path $\paths$ in $\Kripke$ such that $\paths\models\phi$ iff
 there exists a reachable SCC $B$ 
 such that $B$ is accepting for $\phi$.
\end{theorem}
\begin{proof}
\begin{enumerate}
\item[$\Rightarrow$]
Since $\Kripke$ is finite, for any $\paths\in\pathSet^\omega(\Kripke)$, there exists a smallest index $k$, such that all states in $\suffix{\paths}{k}$ will be visited by
  infinite times. By Corollary~\ref{coro:suffixpath}, it suffices to show that $\paths\models\phi$ iff
   $\suffix{\paths}{k}\models\phi$ since
   one can check that $\phi$ is a fairness. For convenience, let $\paths_1 = \suffix{\paths}{k}$.

Let $B_1$ be the set of states on $\paths_1$. Obviously, all states in $B_1$ are connected since all states will be visited by infinite times.
$\paths_1\models\F\G l$ means $s\models l$ for each $s\in B_1$ and $\paths\models\G\F l_j$ means that there exists $s\in B_1$ such that
$s\models l_j$ for each $l_j$. Let $B=B_1\subseteq\states$, then $B$ is an SCC and is accepting for $\phi$.
\item[$\Leftarrow$] This direction is trivial, since we can always construct a path $\paths_2$ that
starts from any $s\in B$ and visits all states in $B$ by infinite times.
Since B is reachable, we can find
a finite path $\paths_1$ which starts from the initial state and reaches the first state of
$\paths_2$. Let $\paths = \paths_1\concat\paths_2$. Obviously $\paths\models\F\G l \land (\bigwedge_{j=1}^{m} \G\F
 l_{j}) $, thus we complete the proof.
\end{enumerate}
\end{proof}

\begin{algorithm}[!t]
    \caption{The procedure $\mathtt{fairMC}$ for
      checking whether $\Kripke\models\phi$, where $\phi$ is a fair
      formula in $\FG$. $\mathtt{fairMC}(\phi,\Kripke)$ returns
      $\mathit{True}$ if $\Kripke\models\phi$, and $\mathit{False}$ otherwise.}\label{alg:fgfg}
    \begin{algorithmic}[1]
       \Procedure{$\mathtt{fairMC}$}{$\phi,\Kripke$}
        \State $\flat(\neg\phi) \equiv \lor_{i=1}^m \phi_i
        = \lor_{i=1}^m(\F\G l_i\land(\bigwedge_{j=1}^{n_i}\G\F l_{i,j}))$;
        \ForAll{($1\le i\le m$)}\label{ln:flat}
            \State $B\leftarrow\{s\in\states\mid s\models l_i\}$;
            \If{$B\neq \emptyset$}
                \ForAll{(SCC $B'\subseteq B$)}
                    \If{{($B'$ is accepting for $\phi_{i}$)}\label{ln:accept}}
                        	\State \Return $\mathit{False}$;
                    \EndIf
                \EndFor
            \EndIf
        \EndFor
   \State \Return $\mathit{True}$;
   \EndProcedure
    \end{algorithmic}
\end{algorithm}
Based on Theorem~\ref{thm:existspath-FG}, Algorithm~\ref{alg:fgfg} describes the procedure to determine
whether all paths in $\Kripke$ satisfy a given fair formula $\phi$
in $\FG$. For this, the algorithm first syntactically transforms
$\neg\phi$ into an equivalent formula of the form
$ \bigvee_{i=1}^m \F\G l_i\land(\bigwedge_{j=1}^{n_i} \G\F l_{i,j})$.
For each $1\le i\le m$, we then try to find an
\emph{accepting SCC} $B'$ such that all states in $B'$ satisfy $l_i$
and at least one state in $B'$ satisfies $l_{i,j}$ for each $1\le j\le
n_i$. In case an accepting SCC is found, there exists a path in
$\Kripke$ violating $\phi$, hence $\Kripke\not\models\phi$; otherwise
we conclude that $\Kripke\models\phi$.

By Theorem~\ref{thm:existspath-FG}, the soundness and completeness of Algorithm~\ref{alg:fgfg} immediately follows.

Note the checking of whether a SCC $B'$ is accepted by $\phi_i$ in
line~\ref{ln:accept} of Algorithm~\ref{alg:fgfg} can be done easily in
linear time with respect to $\size{B'}$.
Let $\size{\Kripke}$ denote
the size of the given model, i.e., the
total number of states and transitions. The complexity of
Algorithm~\ref{alg:fgfg} is shown in the following theorem.
\begin{theorem}\label{thm:fgfg-complexity}
  Algorithm~\ref{alg:fgfg} runs in time $\mathcal{O}(\size{\Kripke}\times 2^{\size{\phi}})$ and in space $\mathcal{O}(\size{\Kripke}+\size{\phi}\times 2^{\size{\phi}})$.
\end{theorem}
Due to case~\ref{push-ant-dis}) in the explanation of Theorem~\ref{thm:normform-fg} and
the transformation that gives a formula of DNF, the resulting formula length can be $\mathcal{O}(2^{\size{\phi}})$ in the worst case.
Suppose $n_1$ is the number of propositional formulas first preceded by $\G$, and
$n_2$ for number of propositional formulas first preceded by $\F$, obviously $n_1 + n_2 \in
\mathcal{O}(\size{\phi})$. We then have $2^{n_1}$ options for $\F\G l$ formulas and
$2^{n_2}$ for $\land_{k} \G\F l_k$ since the number of $l_k$ is $n_2$,
 so we will at most have $2^{n_1 + n_2}$
formulas have the form $\F\G l\land (\bigwedge_{k=1}^{m}\G\F l_{k})$ and each formula
of that form at most has $n_2 + 1$ propositional formulas,
which means that formula length can be $\size{\phi}\times 2^{\mathcal{O}(\size{\phi})}$ in the worst case. That is,
we will at most have $2^{n_1 + n_2}$
formulas with the form $\F\G l\land (\bigwedge_{k=1}^{m}\G\F l_{k})$, and the time for model checking
$\F\G l\land (\bigwedge_{k=1}^{m}\G\F l_{k})$ will be $\size{\Kripke}$ to traverse
all SCCs.
Comparing to the classical algorithm presented
in~\cite{DBLP:conf/lics/VardiW86}, Algorithm~\ref{alg:fgfg} has the same
time complexity. However, experiment shows that our
algorithm achieves much better performance comparing to the classical one.
Furthermore, Algorithm~\ref{alg:fgfg} reduces the space
complexity from $\mathcal{O}(\size{\Kripke}\times 2^{\mathcal{O}(\size{\phi}}))$ to
$\mathcal{O}(\size{\Kripke} + 2^{\mathcal{O}(\size{\phi}}))$ for fairness in $\FG$.

\subsection{Expressiveness of fairness in $\FG$}
\label{subsec:expr-fg}
We have presented an efficient  algorithm to handle the fairness in $\FG$. The question
then arises whether  fair formulas in $\FG$ are expressive enough to encode all fair formulas in $\LTL$?
First,  one can easily verify that the fairness LTL formula $\F\G(a\U b)$ is equivalent to $\F\G(a\lor b)\land\G\F b$.
Intuitively, eventually there is a looping path that satisfies $a\U b$ at every position
is equivalent to that eventually there is a loop path that every state satisfies $a\lor b$ and
there exists at least one state on the loop that satisfies $b$.

The transformation does not work in general. In the following, we show
that $\phi = \F\G(a\lor\X(b \U c))$  can not be expressed by
any fairness in $\FG$. It is easy to see that $\phi$ is a fairness by
Lemma~\ref{lemma:fair-chara}.  But it is impossible to find an
equivalent formula in $\FG$ to represent $\phi$ since the order of
states in SCC matters. We show that $\phi$ can not be
represented as a fairness in $\FG$ by an example in the following.

For the trace $\eta = (\{a\}\{\}\{a,c\})^\omega$ of $\Kripke$ in Example~\ref{ex:kripke},
there are three kinds of letters, namely $\{a\}, \{\}$ and $\{a,c\}$. It
is trivial that $\{a\}\models a$ or $\{a,c\}\models a$. For the word starting from letter $\{\}$, we have
 $\{\}\{a,c\}\cdots\models \X(b \U c)$ since every letter $\{\}$ is directly followed by the letter $\{a,c\}$.
Thus we conclude that $\eta \models \phi$.

By Theorem~\ref{thm:normform-fg}, suppose $\phi\equiv\bigvee_{i=1}^m (\F\G l_i \land (\bigwedge_{j=1}^{n_i} \G\F l_{i,j}))$ holds,
we have $\eta\models\bigvee_{i=1}^m (\F\G l_i \land (\bigwedge_{j=1}^{n_i} \G\F l_{i,j}))$. In other
words, there exists $1\leq i\leq m$ such that $\eta\models \F\G l_i \land (\bigwedge_{j=1}^{n_i} \G\F l_{i,j})$.
Further, we conclude that for any $k\geq 0$, $\suffix{\eta}{k}\models l_i$ and for every $l_{i,j}$, there is at least one out of
letters $\{a\},\{\}$ and $\{a,c\}$ must satisfy $l_{i,j}$.
As a result, $(\{a\}\{a,c\}\{\})^\omega\models \F\G l_i \land (\bigwedge_{j=1}^{n_i} \G\F l_{i,j})$,
which follows that $(\{a\}\{a,c\}\{\})^\omega \models \phi$. Contradiction.

Thus we conclude that fairness in $\FG$ is not powerful enough to express all fairness in $\LTL$.

\subsection{Fairness in $\LTL$}
\label{sec:fairness-properties}
In this subsection we deal with arbitrary fair formulas including
those not expressible in $\FG$.
More notations are needed.
Given $B\subseteq\states$ and $s\in B$,
let $\Kripke^s_B:=(B,s,\tranRel_B,\labels_B)$ where
$\tranRel_B=\tranRel\cap(B\times B)$
and $\labels_B:B\rightarrow 2^\AP$ such that $\labels_B(t)=\labels(t)$
for any $t\in B$. In other words, $\Kripke^s_B$ is a sub-model
of $\Kripke$ where
only states in $B$ and transitions between states in $B$ are
kept.
Moreover, let $\UXLTL$ denote the fragment of
$\LTL$ only containing $\U$ and $\X$ modalities, namely, it is defined
by the following grammar:
\begin{align*}
  \phi ::= a \mid \neg a \mid \phi_1\land\phi_2 \mid \phi_1\lor\phi_2 \mid \X\phi \mid \phi_1\U\phi_2.
\end{align*}
Formulas in $\UXLTL$ are also known as co-safety in literature~\cite{DBLP:journals/fmsd/KupfermanV01,DBLP:conf/icra/BhatiaKV10,DBLP:conf/iros/LacerdaPH14}.

Similar as in Section~\ref{sec:fgfg}, we shall show that any fair
formula can be transformed into an equivalent one, where all
$\U$ and $\X$ modalities can be separated from $\F$ and $\G$ such that the
innermost formulas are all in $\UXLTL$. Such a transformation is
syntactical as well, after which a formula in DNF will be obtained and
moreover, each sub-formula can be handled individually by specific and
efficient algorithms.

\begin{theorem}~\label{thm:u}
Let $\phi\in\LTL$ be a fair formula. Then, it can be transformed into the following equivalent formula, referred to also as its fair normal form:
$$\flat(\phi):=  \bigvee_{i=1}^m \left(\phi_{i0}\land\F\G\phi_{i1}\land(\bigwedge_{j=2}^{n_i} \G\F\phi_{ij})\right)$$
where $\phi_{i0}\in\FG$ and $\phi_{ij}\in\UXLTL$ for all $1\le i\le m$
and $1\le j\le n_i$.
\end{theorem}
\begin{example}\label{ex:fairltl}
Take $\phi=\neg(\F\G(a\lor ( \X(b\U c) \land \F \neg b)))$, then the fair normal form
of $\neg\phi$ is $\F\G a \lor (\F\G(a \lor \X(b\U c)) \land \G\F \neg b)$.
\end{example}

As before, the model checking of an $\LTL$ formula $\phi$ is essentially reduced to the
problem of finding a path in $\Kripke$ satisfying $\flat(\neg\phi)$. Thus, we shall
focus on the procedure of finding a path in $\Kripke$ satisfying the
given formula
$\psi=\phi_0\land\F\G\phi_1\land\G\F\phi_2\land\ldots\land\G\F\phi_n$
with $\phi_0\in\FG$ and $\phi_j\in\UXLTL$ for all $1\le j\le n$. Note by Theorem~\ref{thm:u}, $\flat(\neg\phi)$ is a disjunction of such formulas.

We show how to optimize the procedure of finding a path satisfying $\psi$ or not.
Case $\phi_1\equiv\top$: hence the sub-formula
$\F\G\phi_1$ can be omitted from $\psi$. The formal
procedure for checking whether there exists a path in $\Kripke$
satisfying $\psi$ is presented in Algorithm~\ref{alg:fair}.
As $\psi$ is a fair formula, we can easily show that $\phi_0$
must be also a fair formula. Since $\phi_0\in\FG$, a simple
modification of Algorithm~\ref{alg:fgfg} can be applied to find all accepting SCCs
with respect to $\phi_0$ in $\Kripke$ (line~\ref{ln:fair-acc}). If no
accepting SCC exists, we can terminate, as no path in $\Kripke$ can satisfy $\phi$;
Otherwise, for each accepting SCC $B$ and $\phi_j$
with $2\leq j\leq n$, add a fresh atomic proposition $a_j$ to $A$
(line~\ref{ln:fair-label-update}) iff there exists a state $t\in B$
and $\paths\in\pathSet^\omega(\Kripke^t_B)$ such that $\paths\models\phi_j$
(line~\ref{ln:fair-label}). This step can be done by launching
classical algorithms: A
path $\paths\in\pathSet^\omega(\Kripke^t_B)$ exists such that
$\paths\models\phi_j$ iff $\Kripke^t_B$ does not satisfy $\neg\phi_j$. Finally, an
SCC $B$ is accepted by $\psi$ if at least one
state in $B$ is marked by $a_j$ for each $2\le j\le n$, namely,
$A=\{a_j\}_{2\le j\le n}$ (line~\ref{ln:fair-label-coincide}).

\begin{algorithm}[!t]
    \caption{The procedure $\mathtt{accPath}(\psi,\Kripke)$ for
      checking whether there exists $\paths\in\pathSet^\omega(\Kripke)$
      such that $\paths\models\psi$, where
      $\psi=\phi_0\land\G\F\phi_2\land\ldots\land\G\F\phi_n$ with
      $\phi_0\in\FG$ and $\phi_j\in\UXLTL$ for each $2\le j\le n$.
      $\mathtt{accPath}(\psi,\Kripke)$ returns
      $\mathit{True}$ if a path satisfying $\psi$ is found, and $\mathit{False}$ otherwise.}\label{alg:fair}
    \begin{algorithmic}[1]
       \Procedure{$\mathtt{accPath}$}{$\psi,\Kripke$}
       \State $\mathit{Acc} \leftarrow \{\text{all accepting SCCs with
       respect to }\phi_0\}$;\label{ln:fair-acc}
       \ForAll{($B\in\mathit{Acc}$)}\label{ln:fair-acc-begin}
           \State $A\leftarrow\emptyset$;\label{ln:fair-acc-empty}
           \ForAll{($2\le j\le n$ and $t\in B$)}\label{ln:fair-label-begin}
               \If{(not $(\Kripke^t_B\models\neg\phi_j)$)}\label{ln:fair-label}
                   \State $A\leftarrow A\cup\{a_j\}$;\label{ln:fair-label-update}
               \EndIf
           \EndFor
           \If{($A=\{a_j\}_{2\le j\le n}$)}\label{ln:fair-label-coincide}
              \Return $\mathit{True}$;
           \EndIf
       \EndFor
       \State \Return $\mathit{False}$;
   \EndProcedure
    \end{algorithmic}
\end{algorithm}

The key point behind Algorithm~\ref{alg:fair} is that $\phi_j$ ($2\le j\le
n$) is in $\UXLTL$, the corresponding
B\"uchi automaton of which is terminal~\cite{BloemRS99}. Therefore, once a path $\paths$
satisfies $\phi_j$, we can always find a finite fragment of $\paths$
which suffices to conclude that $\paths\models\phi_j$ regardless of the
remainder of $\paths$. In other words, whenever $\paths\models\phi_j$,
there exists $i\ge 0$ such that
$(\prefix{\paths}{i}\concat\paths')\models\phi_j$ for any infinite path
$\paths'$. Whenever Algorithm~\ref{alg:fair} returns $\mathit{True}$
and finds an accepting $B$ for $\psi$, we can construct a path satisfying $\psi$
as follows:
\begin{enumerate}
\item\label{it:a1}
Let $\paths_1$ be a finite path in $\Kripke^t_B$ for any $t$ such that
all states in $B$ appear in $\paths_1$ for at least once. Traversing all states
in $B$ is useful to witness $\phi_0\in\FG$.
\item\label{it:a2} Continue from the last state of $\paths_1$ and go to a
  state $t_2$ by following any path, where $t_2$ is a state in $B$, from which
  a path satisfying $\phi_2$ exists. Let $\paths'_1$ be the resultant
  path ending at $t_2$. Expand $\paths'_1$ by following the path satisfying $\phi_2$ and
  stop whenever $\phi_2$ is for sure satisfied. Denote the resultant finite
  path by $\paths_2$.
\item Keep extending $\paths_2$ by repeating step~\ref{it:a2} for each
  $3\le j\le n$. Let $\paths_n$ denote the resulting path.
\item Let $\paths'_n$ denote an arbitrary extension of $\paths_n$ such that
   $t$ is a direct successor of the last state of
   $\paths'_n$, namely, $(\paths'_n)^\omega$ is a cyclic path
   in $\Kripke^t_B$.
\end{enumerate}

By construction, it is easy to check that
$(\paths'_n)^\omega\models\psi$, which also shows the soundness and
completeness of Algorithm~\ref{alg:fair}.

Case
$\phi_1\not\equiv\top$: we have to make sure that an accepting path
also satisfies $\F\G\phi_1$. For this purpose, we first transform $\F\G\phi_1$ to a
B\"uchi automaton, denoted $\mathcal{A}_1$, and then build a product
model $\Kripke\times\mathcal{A}_1$ as in the classical
algorithm. Let $a_1$ be a fresh atomic proposition such that $a_1$
holds at a state iff the state is accepting in $\Kripke\times\mathcal{A}_1$.
The remainder of the procedure is similar as the case when
$\phi_1\equiv\top$ except $\F\G\phi_1$ is replaced by $\G\F a_1$
in $\psi$ and the model under checked will be $\Kripke\times\mathcal{A}_1$.
\begin{example}\label{ex:fairltlcheck}
Consider to verify $\phi$ from Example~\ref{ex:fairltl} over $\Kripke$ in Example~\ref{ex:kripke}. As we already have the fair normal form for $\neg\phi$ by Example~\ref{ex:fairltl}, we need to check whether there is a path $\paths$ such that $\paths\models \F\G a$ or $\paths\models\F\G(a \lor \X(b\U c)) \land \G\F \neg b)$. Note that if we first check $\F\G a$, then we employ Algorithm~\ref{alg:fgfg} and terminate here with a counterexample $(s_0)^\omega$.

To further illustrate the algorithm, we
 continue with formula $\F\G(a \lor \X(b\U c)) \land \G\F \neg b$. Since $a\lor\X(b\U c)\not\equiv\top$,
we construct an automaton $\mathcal{A}$ for $\F\G(a\lor\X(b \U c))$ with B\"uchi accepting
condition $\G\F accepting$ where $accepting$ is a new atomic proposition. Then we construct the product of $\Kripke$ and $\mathcal{A}$ and find an SCC accepted by $\G\F accepting \land \G\F \neg b$, in this case, say $\{s_0, s_1, s_2\}$. We
therefore construct a counterexample $(s_0 s_1 s_2)^\omega$. Detailed information of $\mathcal{A}$ and the product
can be found in the appendix.
\end{example}
\paragraph{Discussions}
As mentioned before,
formulas in $\UXLTL$ are guarantee properties according to the
classification in~\cite{CernaP03}. Their corresponding B\"uchi
automata are terminal, for which specific and efficient algorithms
exist~\cite{BloemRS99}. By separating a fair formula, we can
identify sub-formulas belonging to different fragments, each of which
will be handled by specific and efficient algorithms.

\subsection{General Formulas with Fairness Assumptions}
\label{sec:fg}
In this subsection we show how the model checking problem for general $\LTL$
formulas with fairness assumptions
can be accelerated by the specific
algorithms for fair formulas introduced in the above subsections.

Given a fair
formula $\phi_f$ and an $\LTL$ formula $\phi$, the model checking
problem of $\phi$ under the assumption $\phi_f$ reduces to
checking whether $\Kripke\models(\phi_f\implies\phi)$. In order to
make use of our specific algorithm for fairness, the
procedure can be divided into two steps:
\begin{enumerate}
\item $\neg\phi$ is first transformed into a
B\"uchi automaton, denoted $\mathcal{A}_{\neg\phi}$, and the product
of $\mathcal{A}_{\neg\phi}$ and $\Kripke$ is then constructed, where all accepting
states are marked by a fresh atomic proposition $\mathit{accepting}$;
\item Then
$\Kripke\models(\phi_f\implies\phi)$ iff there is no path in the product
satisfying $\phi_f\land\G\F\mathit{accepting}$. Note $\phi_f\land\G\F\mathit{accepting}$ is still a
fair formula, for which our efficient algorithm can be applied.
\end{enumerate}
Note that we can specify some fairness assumption like $\F\G(a\lor ( \X(b\U c) \land \F \neg b))$ in Example~\ref{ex:fairltl} which is not in $\FG$. Moreover, by making use of our algorithm for fairness, we
gain some speed up in the model checking procedure if we choose to check $\F\G a$ in the fair normal form as
discussed in Example~\ref{ex:fairltlcheck}.
\subsection{Formula Characterization}
In this section, we specify some formula sets which are favourable to our algorithm
as well as some formula sets for which our syntactic transformation leads to dramatic blow up of
the formula lengths.

We first characterize some formula sets to which applying our transformation does not lead to
dramatic growth of formula length, and we call them the \emph{fast $\LTL$} formulas.
\begin{definition}\label{def:fast1-chara}
  Let $\Sigma_f$ be a subset of $\LTL$ formulas which is constructed by following rules. Then $ \phi_{f}, \phi_{e}\in\Sigma_f$
  where $\phi_{1} \in \UXLTL$.
  \begin{align*}
  \phi_{0} &::= \phi_1 \mid \F\phi_{0} \mid \G \phi_{0} \mid  \phi_{0} \land \phi_{0} \\
   \phi_{f} &::= \phi_{0} \mid \phi_{f} \lor \phi_{f} \\
   \phi_{e} &::= \phi_1 \mid \phi_{e} \lor \phi_{e} \mid \phi_{e} \land \phi_{e} \mid \F \phi_{e} \mid \G \phi_{e} \\
 \end{align*}
\end{definition}
By induction on the structure of formulas defined in Definition~\ref{def:fast1-chara} and
similar analysis from Theorem~\ref{thm:fgfg-complexity}, it is straightforward to show that:
\begin{corollary}\label{coro:linear-growth}
Let $\phi_{f}$($\phi_e$) be a formula defined in Definition~\ref{def:fast1-chara} and $\phi'_{f}$
($\phi'_e$) be the resulting formula after the transformation defined in Theorem~\ref{thm:u}.
Then $\len{\phi'_{f}} = \mathcal{O}(\len{\phi_{f}})$.
Similarly, we have $\len{\phi'_{e}} = \mathcal{O}(2^{\len{\phi_{e}}})$.
\end{corollary}

In the following, we give the intuition why the transformation increase the formula length by the following
example.
\begin{example}\label{ex:lenBlowup}
Let
\begin{align*}
 \phi = \psi_1\U\psi_2 = &((\G\F a_1 \land \G\F a_2) \lor \cdots \lor (\G\F a_{p-1}\land\G\F a_{p})) \\
        &\U ((\G\F b_1 \lor \G\F b_2) \land \cdots \land (\G\F b_{q-1}\lor\G\F b_{q}))
\end{align*}
Clearly, $\len{\phi}=\mathcal{O}(p+q)$. We need first get all $\F$ and $\G$ modalities out of
the scope of $\U$. To this end, by rules of $(\phi_1\land\phi_2)\U\phi_3\equiv\phi_1\U\phi_3\land\phi_2\U\phi_3$ and
 $\phi_1\U(\phi_2\lor\phi_3)\equiv\phi_1\U\phi_2\lor\phi_1\U\phi_3$, it requires us
 to transform $\psi_1$ to CNF form and $\psi_2$ to DNF form. After that, we get
 a formula which is of size $\mathcal{O}(2^{\len{\phi}})$.
\end{example}

We remark that our transformation does not work when the formula contains $\W$ modalities,
so we replace $\W$ with $\G$ and $\U$ modalities. As a result,
it may increase the number of modalities after negating a formula. Take $\phi=\F\G(\neg a \lor (\neg b \U \neg c))$ for example, after negating $\phi$, it gives us $\G\F (a \land ((\neg b\land c)\W (b\land c)))$, which is equivalent to $\G\F (a \land (\G (\neg b\land c)\lor ((\neg b\land c)\U (b\land c))))$.
After applying the formula transformation, the resulting formula becomes
$(\F\G (\neg b\land c)\land\G\F a)  \lor \G\F(a\land ((\neg b\land c)\U(b\land c)))$. We notice that the reduction for $\W$ modality
 contributes to the growth of the formula length.
\section{Experiment}
\label{sec:experiment}

In this section we first illustrate briefly how our algorithm is
implemented symbolically in NuSMV and then compare the experiment results
with existing algorithms.
NuSMV is a Symbolic Model Verifier extending the first
BDD-based model checker
SMV~\cite{DBLP:journals/iandc/BurchCMDH92}. Compared to tools based
on explicit representations, NuSMV is able to handle relatively more
complex formulas~\cite{RozierM2007}, which is the main reason for
choosing NuSMV in our experiment.

We implement our algorithm in NuSMV symbolically. The algorithm
first decomposes a given formula syntactically to the specific form
according to Theorems~\ref{thm:normform-fg} and \ref{thm:u} and then uses the \emph{fair cycle detection
algorithm} proposed by Emerson and Lei~\cite{DBLP:conf/lics/EmersonL86}
to find accepting SCCs. For instance, given a fair formula $\phi\in\FG$ such that
$\flat(\phi)=\bigvee_{i=1}^m \left(\F\G l_i \land (\bigwedge_{j=1}^{n_i}
  \G\F l_{i,j})\right)$, the fair cycle detection algorithm can be applied
to determine whether there exists an SCC in $\Kripke$ satisfying $\F\G
l_i \land(\bigwedge^{n_i}_{j=1}\G\F l_{i,j})$ for some $1\le i\le m$.
By doing so, we avoid enumerating all SCCs one by one.

We adopt two well-known and scalable problems as our benchmarks:
dining philosopher problem (PD) and binary semaphore protocol
(BS). Their sizes are summarized in Table~\ref{tab:size}, where
``Size'' refers to the number of reachable states for each model,
PD$x$ denotes the PD model with $x$ philosophers, and similarly for BS$x$.
All experiment results were obtained on a computer with an Intel(R) Core(TM) i7-2600
3.4GHz CPU running Ubuntu 14.04 LTS. We set time and memory limits
to be 2 hours and 3 GB, respectively. The source code and several cases can be downloaded
from

\emph{\url{http://iscasmc.ios.ac.cn/?page_id=984}}
\begin{table}[tp]
\centering
\caption{Number of reachable states}\label{tab:size}
\begin{tabular}{|c|c|c|c|c|c|c|c|}\hline\addtolength{\tabcolsep}{200pt}
Model  &  PD6 & PD9 & PD12 & BS4 & BS8 & BS12 & BS16 \\ \hline
Size  & 566 & 13605 & 324782 & 80 & 2304 & 53248 & 1114110\\
\hline
\end{tabular}
\end{table}

We consider three
categories of formulas.

\subsection{Fair $\FG$ formulas}
 The first category takes formulas often used
in verification tasks.
Specifically, for PD model we consider the following
formula, saying that the first philosopher will eat eventually if
no one will be starved (fairness assumption), namely, whenever a
philosopher is ready, he/she will be able to eat eventually:
$$
\mathit{Spec}_1=\left(\bigwedge_{i=1}^n (\G\F\mathit{ready}_i \implies \G\F\mathit{eat}_i)\right)\implies\F\mathit{eat}_1
$$
For BS model, we consider the following two formulas:
$$
\begin{array}{rcl}
  \mathit{Spec}_2 &=&  \left(\bigwedge\limits_{i=1}^n
(\G\F\mathit{enter}_i \implies \G\F\mathit{critical}_i)\right)\\
&& \implies\F\mathit{critical}_1 \\
\mathit{Spec}_3 &=& \left(\bigwedge\limits_{i=1}^n (\G\F \mathit{enter}_i \implies
\G\F \mathit{critical}_i)\right) \\
&& \implies\left((\neg\mathit{critical}_1\land \neg
  \mathit{critical}_3) \U \mathit{critical}_2\right)
\end{array}
$$
$\mathit{Spec}_2$ denotes a similar specification as
$\mathit{Spec}_1$, while $\mathit{Spec}_3$
requires that the second process entering the critical part before the first and
third processes. Notice that all given fairness assumptions are simple formulas in $\FG$.

In the following we write NuSMV to represent the automata-theoretic
approach implemented in NuSMV. Table~\ref{tab:spec} shows both the
time and memory spent by our algorithm and NuSMV to check formulas in
the first category on PD and BS models, where T-O and M-O denote
``timeout'' and ``out-of-memory'', respectively.

The above assumptions in formulas $\mathit{Spec}_i$ ($i=1,2,3$) are fair $\FG$ formulas.
In this case, our algorithm avoids the product construction entirely. From
Table~\ref{tab:spec}, we can see that our algorithm outperforms NuSMV
in almost all cases. In particular, our algorithm terminates in
seconds for some cases, while NuSMV runs out of time or memory.
\begin{table}[!t]
 \caption{Time (second) and memory usage (MB) for formulas in the
   first category}\label{tab:spec}
\centering
  \scalebox{1}{\begin{tabular}{cccccc}
    \toprule
    \multicolumn{1}{c}{\multirow{2}{*}{Formula}}
    & \multicolumn{1}{c}{\multirow{2}{*}{Model}}
    & \multicolumn{2}{c}{Time (second)}
    & \multicolumn{2}{c}{Memory (MB)} \\
    \cmidrule{3-6}
    & & Ours & NuSMV
    & Ours & NuSMV   \\ \hline

    \multirow{3}{*}{$\mathit{Spec}_1$}
    & PD6 & \textbf{2.65} & 19.49 & \textbf{63.90} & 136.98 \\
    & PD9 & \textbf{1373.41} & T-O & \textbf{113.07} & T-O \\
    & PD12 & T-O & T-O & T-O & T-O \\ \hline
    \multirow{4}{*}{$\mathit{Spec}_2$}
    & BS4 & 0.30 & \textbf{0.15} & \textbf{11.97} & 20.67 \\
    & BS8 & \textbf{0.04} & 172.63 & \textbf{14.01} & 141.04 \\
    & BS12 & \textbf{0.11} & T-O & \textbf{33.85} & T-O \\
    & BS16 & \textbf{1.06} & M-O & \textbf{319.58} & M-O \\ \hline
    \multirow{4}{*}{$\mathit{Spec}_3$}
    & BS4 & \textbf{0.02} & 0.10 & \textbf{12.22} & 18.67 \\
    & BS8 & \textbf{0.03} & 101.25 & \textbf{14.82} & 136.64 \\
    & BS12 & \textbf{0.12} & T-O & \textbf{36.31} & T-O \\
    & BS16 & \textbf{1.14} & M-O & \textbf{337.55} & M-O \\
    \bottomrule
  \end{tabular}}
\end{table}

\begin{table}[!t]
\caption{Formulas in scalable patterns generated by ``genltl''.}
\label{tab:formula-genltl}
\centering
\scalebox{1}{\begin{tabular}{|c|c|c|} \hline
Pattern & \text{genltl} arguments & Formula\\ \hline
p1 & $ \text{--and-fg = $n$} $ & $\land_{i=1}^n \F\G a_i $ \\ \hline
p2 & $ \text{--and-gf = $n$}$ & $\land_{i=1}^n \G\F a_i $  \\ \hline
p3 & $ \text{--gh-r = $n-1$} $ &  $\land_{i=1}^{n-1} (\G\F a_i \lor \F\G a_{i+1}) $  \\ \hline
p4 & $ \text{--ccj-xi = $n$, --or-fg = $n$}$ & $ \lor_{i=1}^n \F\G a_i $\\ \hline
\end{tabular}}
\end{table}

\subsection{Fair Pattern Formulas}
We
consider the second category of fair formulas generated by ``genltl''
-- a tool of Spot library~\cite{Duret04MAS} to generate formulas of
scalable patterns. These patterns and
sample formulas are presented in Table~\ref{tab:formula-genltl}, where
column ``\text{genltl} arguments'' denotes arguments used by ``genltl'' to generate
corresponding formulas and $n$ the
number of philosophers in PD or the number of processes in BS. In
Table~\ref{tab:formula-genltl} and the following formulas, we use
$a_i,b_i,\ldots$ as placeholders which will be replaced by proper
atomic propositions during the experiment. To ease the presentation,
we omit the details here. The time and memory usages of our algorithm
and NuSMV to model check formulas in Table~\ref{tab:formula-genltl}
are presented in Figure~\ref{fig:genltl:1} where
we mark by circles and triangles the running time
and maximal memory consumption respectively. Each circle (triangle)
corresponds to the time (memory) consumption of our algorithm and NuSMV.
The coordinate values of the $y$ axis and $x$ axis are the corresponding
experimental results for
NuSMV and our algorithm respectively.
We fill the marks with
red color when it runs out of time and with blue color for memory out.
For all cases, our algorithm consumes a negligible
amount of time and memory comparing to NuSMV, which runs out of time
and memory in many cases.
All points above the main diagonal indicate that our algorithm is faster or consumes less memory
than NuSMV, which is the case for all large examples. Moreover, we tried Spin~\cite{Holzmann1997MCS} for generating the automata
for formulas in Table~\ref{tab:formula-genltl}, it can not return the answer within 30 minutes for a single formula. We note
that we have run experimental results on more generated pattern formulas and observe very similar results as the one presented here.
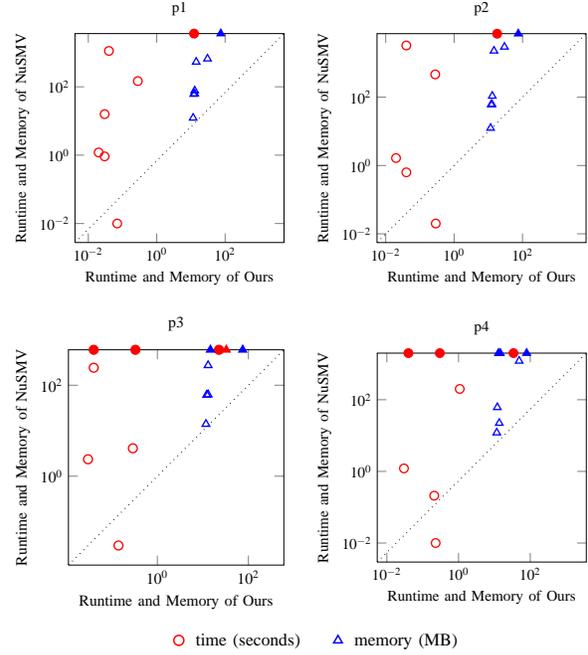
\begin{figure}
  \centering
\resizebox{340pt}{!}{
\begin{tikzpicture}
\scriptsize
\pgfplotsset{
	scatter/classes={
		time={mark=o,solid,semithick,colortime},%
		memo={mark=triangle,solid,semithick,colormem},%
		timeTOut={mark=*,fill=red!20,semithick,colortime},%
		timeMOut={mark=*,fill=blue!20,semithick,colortime},%
		memoMOut={mark=triangle*,fill=colormem,semithick,draw=colormem},
		memoTOut={mark=triangle*,fill=colortime,semithick,draw=colortime}%
	}
	}

	\path[use as bounding box] (-1.9,2.1) rectangle (10.5,-7.1);
	\node (p1pro) at (0,0) {%
	\resizebox{40mm}{!}{
	\begin{tikzpicture}
	\begin{loglogaxis}[
		legend to name=fgltl-plot-legend,
		legend columns=6,
		legend entries={\kern-3mm time (seconds), \kern-3mm memory (MB)},
		scatter,
		title={p1},
		ylabel={Runtime and Memory of NuSMV},
		xlabel={Runtime and Memory of Ours},
		unit vector ratio*=1 1 1,
		unit rescale keep size=true,
		xmax=3600,ymax=3600,
		width=50mm,
		height=50mm
		]
		\addplot[scatter, only marks, scatter src=explicit symbolic] coordinates {
			(0.02, 1.20) [time]
			(0.28, 146.88) [time]
			(12.35, 3600) [timeMOut]
			(0.07, 0.01) [time]
			(0.03, 0.92) [time]
			(0.03, 15.85) [time]
			(0.04, 1119.98) [time]
			
			(12.91, 77.21) [memo]
			(30.53, 661.09) [memo]
			(74.54, 3600) [memoMOut]
			(11.62, 12.29) [memo]
			(12.06, 61.09) [memo]
			(12.86, 63.07) [memo]
			(14.26, 533.50) [memo]
		};
		\draw[dotted] (rel axis cs:0,0) to (rel axis cs:1,1);
	\end{loglogaxis}
	\end{tikzpicture}
	}};
	
	\node (p2pro) at ($(p1pro) + (4.25,0)$) {%
	\resizebox{40mm}{!}{
	\begin{tikzpicture}
	\begin{loglogaxis}[
		scatter,
		title={p2},
		ylabel={Runtime and Memory of NuSMV},
		xlabel={Runtime and Memory of Ours},
		unit vector ratio*=1 1 1,
		unit rescale keep size=true,
		xmax=7200,ymax=7200,
		width=50mm,
		height=50mm
		]
		\addplot[scatter, only marks, scatter src=explicit symbolic] coordinates {
				(0.04, 0.63) [time]
				(0.28, 463.51) [time]
				(18.08, 7200) [timeMOut]
				(0.29, 0.02) [time]
				(0.02, 1.65) [time]
				(0.04, 3255.54) [time]
				
				(12.80, 61.61) [memo]
				(29.59, 2960.68) [memo]
				(74.96, 7200) [memoMOut]
				(11.64, 12.45) [memo]
				(12.11, 60.93) [memo]
				(13.04, 108.27) [memo]
				(14.63, 2268.04) [memo]
		};
		\draw[dotted] (rel axis cs:0,0) to (rel axis cs:1,1);
	\end{loglogaxis}
	\end{tikzpicture}
	}};
	
	\node (p3pro) at ($(p1pro) + (0,-4.5)$) {%
	\resizebox{40mm}{!}{
	\begin{tikzpicture}
	\begin{loglogaxis}[
		scatter,
		title={p3},
		ylabel={Runtime and Memory of NuSMV},
		xlabel={Runtime and Memory of Ours},
		unit vector ratio*=1 1 1,
		unit rescale keep size=true,
		xmax=600,ymax=600,
		width=50mm,
		height=50mm
		]
		\addplot[scatter, only marks, scatter src=explicit symbolic] coordinates {
		     (0.29, 4.11) [time]
		     (0.33, 600) [timeTOut]
		     (22.55, 600) [timeMOut]
		     (0.14, 0.03) [time]
		     (0.03, 2.36) [time]
		     (0.04, 242.61) [time]
		     (0.04, 600) [timeMOut]
		
		     (13.02, 62.15) [memo]
		     (32.60, 600) [memoTOut]
		     (74.89, 600) [memoMOut]
		     (11.68, 13.98) [memo]
		     (12.15, 61.41) [memo]
		     (13.02, 273.48) [memo]
		     (14.64, 600) [memoMOut]
		};
		\draw[dotted] (rel axis cs:0,0) to (rel axis cs:1,1);
	\end{loglogaxis}
	\end{tikzpicture}
	}};

	\node (p4pro) at ($(p2pro) + (0,-4.5)$) {%
	\resizebox{40mm}{!}{
	\begin{tikzpicture}
	\begin{loglogaxis}[
		scatter,
		title={p4},
		ylabel={Runtime and Memory of NuSMV},
		xlabel={Runtime and Memory of Ours},
		unit vector ratio*=1 1 1,
		unit rescale keep size=true,
		xmax=2000,ymax=2000,
		width=50mm,
		height=50mm
		]
		\addplot[scatter, only marks, scatter src=explicit symbolic] coordinates {
			(0.21, 0.21) [time]
			(1.08, 200.03) [time]
			(34.11, 2000) [timeMOut]
			(0.23, 0.01) [time]
			(0.03, 1.21) [time]
			(0.04, 2000) [timeMOut]
			(0.30, 2000) [timeMOut]
			
			(13.68, 22.29) [memo]
			(49.07, 1209.45) [memo]
			(79.34, 2000) [memoMOut]
			(11.64, 12.03) [memo]
			(12.07, 61.21) [memo]
			(12.97, 2000) [memoMOut]
			(14.58, 2000) [memoMOut]
		};
		\draw[dotted] (rel axis cs:0,0) to (rel axis cs:1,1);
	\end{loglogaxis}
	\end{tikzpicture}
	}};
    \draw ($(p3pro) + (2.4,-2.5)$) node{\ref{fgltl-plot-legend}};
\end{tikzpicture}
}
  \caption{Comparison With NuSMV for Generated Formulas}
  \label{fig:genltl:1}
\end{figure}

We remark that all formulas in Table~\ref{tab:formula-genltl} are
simple formulas, actually a subset of $\FG$, which can be converted to
simple Streett/Rabin fairness conditions. We expect some speedup if
optimisations~\cite{Baier:2008:PMC:1373322} for treating simple
fairness are implemented in NuSMV. Our algorithm for fairness in $\FG$
follows the same idea except that we first conduct a formula
transformation so that we can handle fairness like $\G\F(a\land \G
b)$. More importantly, our treatment of fair $\FG$ formulas is also an
essential preparation step of handling general LTL fair formulas, as considered below.

\begin{table*}[tbh]
\caption{Formula patterns used in our experiment}
\label{tab:formula-cav}
\centering
\scalebox{1}{\begin{tabular}{|c|c|} \hline
Pattern &  Formula\\ \hline
p5 & $\lor_{1\leq i\leq n} ((\F\G a_i \lor \G\F b_i) \land (\F\G c_i \lor \G\F d_i))$ \\ \hline
p6 & $\land_{1\leq i\leq n} ((\F\G a_i \lor \G\F b_i) \land (\F\G c_i
\lor \G\F d_i))$ \\ \hline
p7 & $\land_{1\leq i\leq n}((\G\F (a_i \land \X\X b_i ) \lor \F\G b_i) \land
\F\G (c_i \lor (\X d_i \land \X\X b_i)))$ \\ \hline
p8 & $\lor_{1\leq i\leq n}((\G\F (a_i \land \X\X b_i ) \lor \F\G b_i) \land
\F\G (c_i \lor (\X d_i \land \X\X b_i)))$\\ \hline
p9 & $\land_{1\leq i\leq n}(\F\G (a_i \lor c_i \lor (a_i \U b_i) \lor (c_i
\U d_i))$ \\ \hline
p10 & $\lor_{1\leq i\leq n}(\F\G (a_i \lor c_i \lor (a_i \U b_i) \lor (c_i
\U d_i))$ \\ \hline
p11& $\lor_{1\leq i\leq n} (\F\G (a_i \lor (a_i\U b_i)) \lor \G\F (c_i \land (c_i \U d_i))) $ \\ \hline
p12 & $\land_{1\leq i\leq n} (\F\G (a_i \lor (a_i\U b_i)) \lor \G\F (c_i \land (c_i \U d_i)))$ \\ \hline
p13 & $\land_{1\leq i\leq n} (\F\G((a_i \land \X\X b_i \land \G\F b_i)\U(\G(\X\X \neg c_i \lor \X\X(a_i \land b_i)))))$ \\ \hline
p14 & $\land_{1\leq i\leq n} (\G(\F\neg a_i \land \F(b_i \land \X \neg c_i) \land \G\F(a_i \U d_i)) \land \G\F((\X d_i)\U(b_i \lor \G c_i))) $ \\ \hline
p15 & negations of formulas in p13 \\ \hline
p16 & negations of formulas in p14 \\ \hline
\end{tabular}}
\end{table*}

\setlength{\textfloatsep}{12pt}
\begin{table*}[t]
 \caption{Time Usage (second)}\label{tab:time-cav}
\centering
\scalebox{0.9}{\begin{tabular}{ccccccccccccccc}
    \toprule
    \multicolumn{1}{c}{\multirow{2}{*}{model}}
    & \multicolumn{2}{c}{p5}
    & \multicolumn{2}{c}{p6}
    & \multicolumn{2}{c}{p7}
    & \multicolumn{2}{c}{p8}
    & \multicolumn{2}{c}{p9}
    & \multicolumn{2}{c}{p10}\\
    \cmidrule{2-13}

    & Ours & NuSMV
    & Ours & NuSMV
    & Ours & NuSMV & Ours & NuSMV
    & Ours & NuSMV
    & Ours & NuSMV   \\ \hline
    PD6    & \textbf{0.11} & T-O & \textbf{0.17} & T-O   & \textbf{0.27} & 2216.90 & \textbf{43.65} & T-O  & \textbf{0.49} & 3.88 & 4353.34 & \textbf{3.90}\\
    PD9   &  \textbf{1.09} & M-O & \textbf{0.38} & M-O    & \textbf{1.74} & M-O & T-O & T-O  & \textbf{18.54} & M-O & T-O & M-O\\
    PD12  & \textbf{41.45} & M-O & \textbf{12.79} & M-O   & \textbf{129.50}
    & M-O & T-O & M-O  & \textbf{1344.55} & M-O  & M-O & M-O \\
\hline
    BS4   & \textbf{0.07} & 10.96 & \textbf{0.23} & 26.67    & \textbf{0.08} & 116.65 & \textbf{0.64} & 14.61 & \textbf{0.06} & 0.07  & 6.49 & \textbf{0.06}\\
    BS8   & \textbf{0.14} & M-O & \textbf{0.02} & M-O   &\textbf{0.03} & T-O & \textbf{1072.92} & T-O &  \textbf{0.08} & 5.36  & T-O & \textbf{3.34}\\
    BS12  & \textbf{4.55} & M-O & \textbf{0.25} & M-O   & \textbf{0.05} & M-O & T-O & M-O  & \textbf{0.32} & 2519.18  & M-O & \textbf{243.44} \\
    BS16  & \textbf{377.33} & M-O & \textbf{1.07} & M-O   & \textbf{0.41} & M-O & T-O & M-O & \textbf{0.94} & T-O   & M-O & M-O  \\
    \bottomrule
  \end{tabular}}
\end{table*}

\begin{table*}[tp]
 \caption{Memory Usage (MB) }\label{tab:mem-cav}
\centering
  \scalebox{0.9}{\begin{tabular}{ccccccccccccc}
    \toprule
    \multicolumn{1}{c}{\multirow{2}{*}{model}}
    & \multicolumn{2}{c}{p5}
    & \multicolumn{2}{c}{p6}
    & \multicolumn{2}{c}{p7}
    & \multicolumn{2}{c}{p8}
    & \multicolumn{2}{c}{p9}
    & \multicolumn{2}{c}{p10}\\
    \cmidrule{2-13}

    & Ours & NuSMV
    & Ours & NuSMV
    & Ours & NuSMV  & Ours & NuSMV
    & Ours & NuSMV
    & Ours & NuSMV \\ \hline
    PD6   & \textbf{14.10} & T-O & \textbf{13.26} & T-O    & \textbf{14.63} & 507.14 & \textbf{90.35} & T-O & \textbf{40.30} & 95.27  & 1877.45 & \textbf{159.64} \\
    PD9   & \textbf{55.77} & M-O & \textbf{35.61} & M-O   & \textbf{48.17} & M-O & T-O & T-O & \textbf{73.23} & M-O & T-O & M-O\\
    PD12  & \textbf{98.70} & M-O & \textbf{76.07} & M-O   & \textbf{92.81} & M-O & T-O & M-O & \textbf{141.55} & M-O & M-O & M-O  \\
\hline
    BS4   & \textbf{11.77} & 62.47 & \textbf{11.70} & 63.75  & \textbf{12.01} & 61.32 & \textbf{19.64} & 61.70 & \textbf{13.37} & 20.36 & 89.45 & \textbf{19.57}\\
    BS8   & \textbf{16.44} & M-O & \textbf{12.16} & M-O   & \textbf{12.68} & T-O & \textbf{508.83} & T-O & \textbf{17.21} & 62.23 & T-O & \textbf{61.89}\\
    BS12  & \textbf{74.95} & M-O & \textbf{12.94} & M-O   & \textbf{14.28} & M-O & T-O & M-O  & \textbf{40.95} & 683.38 & M-O & \textbf{394.22}\\
    BS16  & \textbf{423.19} & M-O & \textbf{14.39} & M-O  & \textbf{16.84} & M-O & T-O & M-O & \textbf{63.17} & T-O  & M-O & M-O \\
    \bottomrule
  \end{tabular}}
\end{table*}

\begin{table*}[t]
 \caption{Time Usage (second)}\label{tab:time-cav-u}
\centering
\scalebox{0.9}{\begin{tabular}{ccccccccccccc}
    \toprule
    \multicolumn{1}{c}{\multirow{2}{*}{model}}
    & \multicolumn{2}{c}{p11}
    & \multicolumn{2}{c}{p12}
    & \multicolumn{2}{c}{p13}
    & \multicolumn{2}{c}{p14}
    & \multicolumn{2}{c}{p15}
    & \multicolumn{2}{c}{p16} \\
    \cmidrule{2-13}
    & Ours & NuSMV
    & Ours & NuSMV
    & Ours & NuSMV
    & Ours & NuSMV
    & Ours & NuSMV
    & Ours & NuSMV   \\ \hline
    PD6   & \textbf{4.48} & 120.16 & \textbf{0.26} & 1793.13 & \textbf{0.32} & 122.76 & \textbf{0.07} & T-O  & 125.82 & \textbf{57.35} & \textbf{3.96} & M-O\\
    PD9   & M-O & M-O & \textbf{1.62} & T-O & \textbf{0.95} & M-O & \textbf{0.29} & M-O  & T-O & T-O & \textbf{2912.29} & M-O\\
    PD12  & M-O & M-O  & \textbf{77.01} & M-O & \textbf{72.20} & M-O & \textbf{10.95} & M-O & T-O & M-O  & T-O & M-O\\
\hline
    BS4   & 15.30 & \textbf{0.26}  & \textbf{0.50} & 3.49 & \textbf{0.14} & 13.13  & \textbf{0.01} & 313.46  & \textbf{0.46} & 10.14  & \textbf{0.32} & 218.21\\
    BS8   & T-O & \textbf{280.91}  & \textbf{0.03} & T-O &  \textbf{0.03} & T-O  & \textbf{0.01} & M-O     & \textbf{742.02} & T-O  & \textbf{0.26} & M-O\\
    BS12  & M-O & M-O  & \textbf{0.05} & M-O & \textbf{0.05} & T-O   & \textbf{0.34} & M-O     & T-O & M-O     & \textbf{1.80} & M-O\\
    BS16  & M-O & M-O & \textbf{1.00} & M-O & \textbf{0.07} & M-O   & \textbf{0.38} & M-O     & M-O & M-O     & \textbf{13.16} & M-O\\
    \bottomrule
  \end{tabular}}
\end{table*}

\begin{table*}[t]
 \caption{Memory Usage (MB) }\label{tab:mem-cav-u}
\centering
  \scalebox{0.9}{\begin{tabular}{ccccccccccccc}
    \toprule
    \multicolumn{1}{c}{\multirow{2}{*}{model}}
    & \multicolumn{2}{c}{p11}
    & \multicolumn{2}{c}{p12}
    & \multicolumn{2}{c}{p13}
    & \multicolumn{2}{c}{p14}
    & \multicolumn{2}{c}{p15}
    & \multicolumn{2}{c}{p16}\\
    \cmidrule{2-13}
    & Ours & NuSMV
    & Ours & NuSMV
    & Ours & NuSMV
    & Ours & NuSMV
    & Ours & NuSMV
    & Ours & NuSMV \\ \hline
    PD6   & \textbf{139.28} & 261.5  & \textbf{14.72} & 457.01 & \textbf{13.76} & 90.76  & \textbf{12.85} & T-O & 101.20 & \textbf{93.16}  & \textbf{74.48} & M-O\\
    PD9   & M-O & M-O & \textbf{50.61} & T-O & \textbf{36.09} & M-O & \textbf{29.83} & M-O  & T-O & T-O & \textbf{440.11} & M-O\\
    PD12  & M-O & M-O & \textbf{90.92} & M-O & \textbf{89.87} & M-O & \textbf{79.64} & M-O & T-O & M-O & T-O & M-O\\
\hline
    BS4   & 95.83 & \textbf{40.36} & \textbf{12.04} & 60.96 & \textbf{12.08} & 51.46 & \textbf{11.68} & 465.63 & \textbf{20.53} & 50.18 & \textbf{14.07} & 460.19\\
    BS8   & T-O & \textbf{143.0} & \textbf{12.71} & T-O & \textbf{12.60} & T-O   & \textbf{12.19} & M-O    & \textbf{214.81} & T-O  & \textbf{27.05} & M-O\\
    BS12  & M-O & M-O & \textbf{14.07} & M-O & \textbf{13.70} & T-O   & \textbf{13.03} & M-O    & T-O & M-O     & \textbf{111.46} & M-O\\
    BS16  & M-O & M-O  & \textbf{16.86} & M-O & \textbf{15.42} & M-O   & \textbf{14.50} & M-O    & M-O & M-O     & \textbf{841.90} & M-O\\
    \bottomrule
  \end{tabular}}
\end{table*}
\subsection{General LTL Fairness}
We consider some general fair LTL formulas, summarized in
Table~\ref{tab:formula-cav}. These formulas are often adopted
to evaluate performance of an $\LTL$ model checker or
planner in the literature; see for
instance~\cite{DBLP:conf/cav/SomenziB00,DBLP:conf/concur/EtessamiH00,DBLP:conf/spin/Pelanek07,DBLP:conf/cav/EsparzaK14}.
The time consumption for checking these formulas is
presented in Table~\ref{tab:time-cav} and \ref{tab:time-cav-u}, while
the memory consumption is shown in Table~\ref{tab:mem-cav} and \ref{tab:mem-cav-u}.
From these results we observe similar phenomena as before for most cases
except for ``p10", ``p11", and ``p15", where our algorithm uses more time
and/or memory than NuSMV for certain cases, particularly when ``p10" and PD models are
concerned. We explain such performance differences in details in the following.

As mentioned before, our algorithm relies on syntactical transformations in
Theorems~\ref{thm:normform-fg} and \ref{thm:u}. These transformations can decompose a fair
formula into smaller sub-formulas, whose corresponding B\"uchi
automata are usually much smaller than the automaton of the original
formula. This is the main reason that our algorithm achieves much
better performance than the classical algorithm for most of the instances. However, the syntactic transformations adopted in
Theorem~\ref{thm:normform-fg} and \ref{thm:u} may cause exponential
blow-up for certain cases; for instance formulas whose negations are in
form of ``p10" and ``p11". In order to push all $\F$ and $\G$
modalities in front of  $\U$ modality, our transformation may
need to transform back and forth between CNF and DNF of some formulas,
especially for those formulas where $\F,\G$ and $\U$ are alternatively
nested for many times. Therefore, for such formulas, the syntactic
transformation may be time-consuming and result in formulas of
exponentially longer than the original ones.

We note that many formulas we take from the literature are
characterized by Definition~\ref{def:fast1-chara}, and
transforming the negation of these formulas only leads to a
linear increase in the formula length.
The exceptions are ``p5", ``p8", ``p10", ``p11", and ``p13-p16". It is
worthwhile to mention that even though for formulas such that the
transformations result in
formulas of exponential length, our algorithm is not necessarily
slower than NuSMV, as the corresponding B\"uchi automata may be
exponentially large as well; for instance ``p8" and ``p10". Finally, our
algorithm outperforms NuSMV for ``p14" and its negation ``p16"; it is faster for ``p13" and is only
slightly slower than  NuSMV for its negation ``p15" for one case.

\section{Conclusion}
\label{sec:concl-future-work}
We presented a novel model checking algorithm for formulas in $\LTL$ with fairness assumptions.
Our algorithm does not follow the automata-theoretic approach completely but
tries to decompose a fair formula into several sub-formulas, each of which
can be handled by specific and efficient algorithms. We showed by
experiment that our algorithm in many cases exceeds NuSMV up to several
orders of magnitudes.

\bibliographystyle{abbrv}
\bibliography{bib}

\begin{thebibliography}{10}

\bibitem{alpern1987recognizing}
B.~Alpern and F.~B. Schneider.
\newblock Recognizing safety and liveness.
\newblock {\em Distributed Computing}, 2(3):117--126, 1987.

\bibitem{Baier:2008:PMC:1373322}
C.~Baier and J.-P. Katoen.
\newblock {\em Principles of Model Checking}.
\newblock The MIT Press, 2008.

\bibitem{DBLP:conf/icra/BhatiaKV10}
A.~Bhatia, L.~E. Kavraki, and M.~Y. Vardi.
\newblock Sampling-based motion planning with temporal goals.
\newblock In {\em {ICRA}}, pages 2689--2696, 2010.

\bibitem{BloemRS99}
R.~Bloem, K.~Ravi, and F.~Somenzi.
\newblock Efficient decision procedures for model checking of linear time logic
  properties.
\newblock In {\em CAV}, volume 1633 of {\em LNCS}, pages 222--235. Springer,
  1999.

\bibitem{DBLP:journals/iandc/BurchCMDH92}
J.~R. Burch, E.~M. Clarke, K.~L. McMillan, D.~L. Dill, and L.~J. Hwang.
\newblock Symbolic model checking: $10^{20}$ states and beyond.
\newblock {\em Inf. Comput.}, 98(2):142--170, 1992.

\bibitem{CernaP03}
I.~Cern{\'a} and R.~Pel{\'a}nek.
\newblock Relating hierarchy of temporal properties to model checking.
\newblock In {\em MFCS}, volume 2747 of {\em LNCS}, pages 318--327. Springer,
  2003.

\bibitem{DBLP:journals/sttt/CimattiCGR00}
A.~Cimatti, E.~M. Clarke, F.~Giunchiglia, and M.~Roveri.
\newblock {NUSMV:} {A} new symbolic model checker.
\newblock {\em {STTT}}, 2(4):410--425, 2000.

\bibitem{DBLP:conf/cav/DanieleGV99}
M.~Daniele, F.~Giunchiglia, and M.~Y. Vardi.
\newblock Improved automata generation for linear temporal logic.
\newblock In {\em CAV}, volume 1633 of {\em LNCS}, pages 249--26. Springer,
  1999.

\bibitem{Duret04MAS}
A.~Duret-Lutz and D.~Poitrenaud.
\newblock {SPOT}: an extensible model checking library using transition-based
  generalized {B}\"uchi automata.
\newblock In {\em MASCOTS}, pages 76--83. IEEE Press, 2004.

\bibitem{DBLP:conf/lics/EmersonL86}
E.~A. Emerson and C.~Lei.
\newblock Efficient model checking in fragments of the propositional
  mu-calculus (extended abstract).
\newblock In {\em LICS}, pages 267--278. {IEEE} Computer Society, 1986.

\bibitem{DBLP:conf/cav/EsparzaK14}
J.~Esparza and J.~Kret{\'{\i}}nsk{\'{y}}.
\newblock From {LTL} to deterministic automata: {A} safraless compositional
  approach.
\newblock In {\em {CAV}}, volume 8559 of {\em LNCS}, pages 192--208. Springer,
  2014.

\bibitem{DBLP:conf/concur/EtessamiH00}
K.~Etessami and G.~J. Holzmann.
\newblock Optimizing {B}{\"{u}}chi automata.
\newblock In {\em {CONCUR}}, volume 1877 of {\em LNCS}, pages 153--167.
  Springer, 2000.

\bibitem{Francez:1986:FAI:19247}
N.~Francez.
\newblock {\em Fairness}.
\newblock Springer, 1986.

\bibitem{Gastin2001}
P.~Gastin and D.~Oddoux.
\newblock Fast {LTL} to {B}{\"u}chi automata translation.
\newblock In {\em CAV}, volume 2102 of {\em LNCS}, pages 53--65. Springer,
  2001.

\bibitem{HammerMAM05}
M.~Hammer, A.~Knapp, and S.~Merz.
\newblock Truly on-the-fly {LTL} model checking.
\newblock In {\em TACAS}, volume 3440 of {\em LNCS}, pages 191--205. Springer,
  2005.

\bibitem{Holzmann1997MCS}
G.~J. Holzmann.
\newblock The model checker {SPIN}.
\newblock {\em IEEE Trans. Softw. Eng.}, 23(5):279--295, 1997.

\bibitem{Kini2015tacas}
D.~Kini and M.~Viswanathan.
\newblock {Limit Deterministic and Probabilistic Automata for
  LTL$\setminus$GU}.
\newblock In {\em TACAS}, volume 9035 of {\em LNCS}, pages 628--642, Berlin,
  Heidelberg, 2015. Springer Berlin Heidelberg.

\bibitem{DBLP:journals/fmsd/KupfermanV01}
O.~Kupferman and M.~Y. Vardi.
\newblock Model checking of safety properties.
\newblock {\em Formal Methods in System Design}, 19(3):291--314, 2001.

\bibitem{DBLP:conf/iros/LacerdaPH14}
B.~Lacerda, D.~Parker, and N.~Hawes.
\newblock Optimal and dynamic planning for {M}arkov decision processes with
  co-safe {LTL} specifications.
\newblock In {\em {IROS}}, pages 1511--1516, 2014.

\bibitem{DBLP:dblp_journals/tse/Lamport77}
L.~Lamport.
\newblock Proving the correctness of multiprocess programs.
\newblock {\em IEEE Trans. Software Eng.}, pages 125--143, 1977.

\bibitem{DBLP:conf/spin/Pelanek07}
R.~Pel{\'{a}}nek.
\newblock {BEEM:} benchmarks for explicit model checkers.
\newblock In {\em {SPIN}}, volume 4595 of {\em LNCS}, pages 263--267. Springer,
  2007.

\bibitem{Pnueli1977TLP}
A.~Pnueli.
\newblock The temporal logic of programs.
\newblock In {\em SFCS}, pages 46--57. IEEE Computer Society, 1977.

\bibitem{DBLP:journals/acta/QueilleS83}
J.~Queille and J.~Sifakis.
\newblock Fairness and related properties in transition systems - {A} temporal
  logic to deal with fairness.
\newblock {\em Acta Inf.}, 19:195--220, 1983.

\bibitem{RenaultDKP13}
E.~Renault, A.~Duret-Lutz, F.~Kordon, and D.~Poitrenaud.
\newblock Strength-based decomposition of the property {B{\"u}chi} automaton
  for faster model checking.
\newblock In {\em TACAS}, volume 7795 of {\em LNCS}, pages 580--593. Springer,
  2013.

\bibitem{RozierM2007}
K.~Y. Rozier and M.~Y. Vardi.
\newblock {LTL} satisfiability checking.
\newblock In {\em SPIN}, volume 4595 of {\em LNCS}, pages 149--167.
  Springer-Verlag, 2007.

\bibitem{Sickert2016cav}
S.~Sickert, J.~Esparza, S.~Jaax, and J.~Kret{\'{\i}}nsk{\'{y}}.
\newblock {Limit-Deterministic B\"uchi Automata for Linear Temporal Logic}.
\newblock In {\em CAV}, LNCS, 2016.
\newblock To appear.

\bibitem{DBLP:journals/fac/Sistla94}
A.~P. Sistla.
\newblock Safety, liveness and fairness in temporal logic.
\newblock {\em Formal Asp. Comput.}, 6(5):495--512, 1994.

\bibitem{Sistla:1985:CPL:3828.3837}
A.~P. Sistla and E.~M. Clarke.
\newblock The complexity of propositional linear temporal logics.
\newblock {\em J. ACM}, 32(3):733--749, 1985.

\bibitem{DBLP:conf/cav/SomenziB00}
F.~Somenzi and R.~Bloem.
\newblock Efficient b{\"u}chi automata from ltl formulae.
\newblock In {\em CAV}, volume 1855 of {\em LNCS}, pages 248--263. Springer,
  2000.

\bibitem{DBLP:conf/lics/VardiW86}
M.~Y. Vardi and P.~Wolper.
\newblock An automata-theoretic approach to automatic program verification
  (preliminary report).
\newblock In {\em LICS}, pages 332--344. IEEE Computer Society, 1986.

\end{thebibliography}

\newpage
\newpage
\appendix
In the appendix we provide correctness proofs for our transformation. While the intuitive idea is straightforward, the proofs are quite technical due to the many cases.
As a preparation step, we   define first another semantic
equivalence between $\LTL$ formulas with restricted to cyclic
sequences as follows:
\begin{definition}\label{def:equiv-c}
Given two formulas $\phi,\psi$, we write $\phi~\equiv_c~\psi$
iff for any $\sigma\in(2^\AP)^*$, $\sigma^\omega\models\phi$ iff $\sigma^\omega\models\psi$.
\end{definition}
We call equivalence relation $\equiv_c$ \emph{cyclic equivalence} since we only care about
the cyclic words.
It is easy to show that $\equiv$ and $\equiv_c$ make no difference for fairness according to
Corollary~\ref{coro:suffixpath}. However, $\equiv_c$ is
still necessary, as along the transformation some non-fair formulas may be generated,  which may only preserve $\equiv_c$ but not $\equiv$.
For instance, $\F\G\phi\equiv_c\G\phi$ holds while $\F\G\phi\equiv\G$ does not
in general.

\subsection{Proof of equations to explain Theorem~\ref{thm:normform-fg}}
\label{proof:equation}
We first give the proof of the equations listed to explain the intuition behind Theorem~\ref{thm:normform-fg}. In the following, we only prove the first equation, the second equation immediately follows
by negating both sides of the first equation.
\begin{enumerate}
\item $\G\F(\phi_1\lor\phi_2)\equiv \G\F\phi_1\lor\G\F\phi_2$ and $\F\G(\phi_1\land\phi_2)\equiv\F\G\phi_1\land\F\G\phi_2$. The direction from $\G\F\phi_1\lor\G\F\phi_2$
to $\G\F(\phi_1\lor\phi_2)$ is trivial. Consider the other direction, for any infinite word
$\sigma$ such that $\sigma\models\G\F(\phi_1\lor\phi_2)$, there exists infinite $j\geq 0$
such that $\suffix{\sigma}{j}\models \phi_1\lor\phi_2$, which implies at least one formula out of
$\phi_1$ and $\phi_2$ will be satisfied infinitely often.

\item $\G\F(\phi_1\land\F \phi_2)\equiv\G\F\phi_1\land\G\F\phi_2$ and
$\F\G(\phi_1\lor\G\phi_2)\equiv\F\G\phi_1\lor\F\G\phi_2$. Obviously, $\sigma\models \G\F (\phi_1 \land \F \phi_2)$ implies
  $\sigma\models\G\F \phi_1 \land \G\F \phi_2$ for any infinite word $\sigma$. For the other direction,
  $\sigma\models\G\F \phi_1 \land \G\F \phi_2$ implies that there
  exist infinitely many $i$ with $i\geq 0$ such that $\suffix{\sigma}{i}\models\phi_1$.
  For every such $i$ from above, since $\sigma\models\G\F\phi_2$, we have $\suffix{\sigma}{i}\models\phi_1\land\F\phi_2$, hence
  $\sigma \models\G\F(\phi_1\land\F\phi_2)$.
\item $\G\F(\phi_1\land\G\phi_2)\equiv\G\F\phi_1\land\F\G\phi_2$ and
$\F\G(\phi_1\lor\F\phi_2)\equiv\F\G\phi_1\lor\G\F\phi_2$.
The direction that $\G\F (\phi_1 \land \G \phi_2)$ implies
  $\G\F\phi_1 \land \F\G \phi_2$ is straightforward. Now let
  $\sigma\models\G\F\phi_1\land\F\G\phi_2$. Obviously,
  $\sigma\models\F\G\phi_2$, which indicates that there exists some $j\geq 0$ such
  that for every $i\geq j$, we have $\suffix{\sigma}{i}\models\G\phi_2$. In addition,
  since $\sigma\models\G\F\phi_1$, we can find infinitely many $k\geq j$, such that
  $\suffix{\sigma}{k}\models \phi_1\land\G\phi_2$.
Therefore, $\sigma\models\G\F(\phi_1\land\G\phi_2)$.
\end{enumerate}
\subsection{Proof of Theorem~\ref{thm:normform-fg}}
\begin{proof}

In the normal form, there are two kinds of modalities, namely $\F\G$
 and $\G\F$. When we consider fairness, equivalence relation $\equiv$ and $\equiv_c$ coincide,
 we therefore use $\G$ and $\F$ to represent $\F\G$ and $\G\F$ respectively in the following
 flatten operation. We give the rules for flatten operation which we use to transfer any formula in $\FG$ to a formula of the normal form. We then prove
that all transformation rules are sound. Formally, we define the flatten operator $\flat$ inductively as follows:
  \begin{enumerate}
  \item\label{flat:literal} $\flat(l) = l$,
  \item\label{flat:disj} $\flat(\phi_1\lor\phi_2) = \flat(\phi_1)\lor\flat(\phi_2)$,
  \item\label{flat:conj} $\flat(\phi_1\land\phi_2)=\flat_{\mathit{dnf}}(\flat(\phi_1)\land \flat(\phi_2))$,
  \item\label{flat:f} $\flat(\F\phi)=\flat_\F(\flat(\phi))$,
  \item\label{flat:g} $\flat(\G\phi)=\flat_{\mathit{dnf}}(\flat_{\G}(\flat_{\mathit{cnf}}(\flat(\phi))))$,
  \end{enumerate}
  where $\flat_{\mathit{dnf}}$ and $\flat_{\mathit{cnf}}$ denote
  transformations to equivalent formulas in disjunctive norm form
  (DNF) and conjunctive norm form (CNF), respectively, and
 $$
 \begin{array}{l}
   \flat_\F(l)=\F l, \flat_\F(\F
   l)=\F l, \flat_\F(\G l)=\G l,\\
   \flat_\G(l)=\G l, \flat_\G(\F
   l)=\F l, \flat_\G(\G l)=\G l,\\
      \flat_\F(\phi_1 * \phi_2)=\flat_\F(\phi_1) *
   \flat_\F(\phi_2)\text{ with } *\in\{\land,\lor\},\\
   \flat_\G(\phi_1 * \phi_2)=\flat_\G(\phi_1) * \flat_\G(\phi_2) \text{ with } *\in\{\land,\lor\}.\\
 \end{array}
 $$
Note that if $\phi_1*\phi_2$ is a propositional formula, we consider it as one formula so that
we do not apply $\flat_\F$ or $\flat_\G$ to $\phi_1$ and $\phi_2$ individually.
As a result, we obtain a formula in form of $\bigvee_{i=1}^m \left(\G l_i \land
(\bigwedge_{j=1}^{n_i} \F l_{i,j})\right)$. Since $\equiv_c$ and $\equiv$ coincide
for fairness, we actually get an equivalent formula in the form of
$\bigvee_{i=1}^m \left(\F\G l_i \land  (\bigwedge_{j=1}^{n_i} \G\F l_{i,j})\right)$.
We notice that rule~\ref{flat:g} is more involved than other rules.
Intuitively, after applying $\flat$ operator to $\phi$,
it gives us a formula, say $\phi'$, which
is in DNF. Further, since $\G$ is not distributive over $\lor$ operator, we have to use
$\flat_{\mathit{cnf}}$ to transform $\phi'$ to a formula in CNF, say $\phi''$.
By applying operator $\flat_\G$ to $\phi''$, we are able to push $\G$ inside and
then get a formula in DNF through $\flat_{\mathit{dnf}}$.

Next, we shall prove that rules~\ref{flat:literal} to \ref{flat:g} are
sound. For this, it suffices to prove the following rules, where
$\phi_1, \phi_2,$ and $\phi$ are arbitrary $\LTL$ formulas.
  \begin{enumerate}
  \item\label{cyc-eq} $\F\G\phi \equiv_c \G\phi$, $\G\F\phi\equiv_c\F\phi$. Trivial that $\G\phi$ implies $\F\G\phi$. We show that
    whenever $\sigma^\omega\models\F\G\phi$, it is also the
    case that $\sigma^\omega\models\G\phi$. This is also straightforward,
    as $\sigma^\omega\models\F\G\phi$ indicates any suffix of
    $\sigma^\omega$ satisfies $\phi$. By negating both sides of the first equation, the second equation follows immediately.
  \item\label{cyc-fandf} $\F (\phi_1 \land \F \phi_2) \equiv_c \F \phi_1 \land \F \phi_2$,
        $\G (\phi_1 \lor \G \phi_2) \equiv_c \G \phi_1 \lor \G
        \phi_2$. Above equations can be proved by the equations in subsection~\ref{proof:equation}
        and together with $\F\G\equiv_c\G$ and $\G\F\equiv_c\F$.
  \item\label{cyc-fandg} $\F (\phi_1 \land \G \phi_2) \equiv_c \F \phi_1 \land \F\G \phi_2$,
        $\G (\phi_1 \lor \F \phi_2) \equiv_c \G \phi_1 \lor \G\F \phi_2$. Those equations can be proved by the equations in subsection~\ref{proof:equation}
        and together with $\F\G\equiv_c\G$ and $\G\F\equiv_c\F$.
 \item $\F\G \phi_1 \land \F\G\phi_2 \equiv \F\G
   (\phi_1\land\phi_2)$. It has been proved before.
  \end{enumerate}
This completes the proof.
\end{proof}
Intuitively, the flatten operator $\flat$ takes a fair formula $\phi\in\FG$ as an input and
outputs a formula in DNF, where each sub-formula is a conjunction of formulas
in form of $l$, $\F l$ or $\G l$.
 We illustrate the definition of $\flat$ operator via an
example as follows:
\begin{example}\label{ex:flatandfg}
Let $\phi=\F\G(a \lor \F b)$. We show how to flatten $\phi$ step by step,
 where numbers above $=$ denote the corresponding rules in the above proof.
  $$
  \begin{array}{rcl}
   \flat(\F b) \stackrel{\ref{flat:literal},\ref{flat:f}}{=}\F b &\qquad&
   \flat(a\lor\F b) \stackrel{\ref{flat:disj}}{=} a\lor\F b \\
   \flat(\G(a\lor\F b))&\stackrel{\ref{flat:g}}{=}& \G a\lor\F b \\
   \flat(\F\G(a\lor\F b))&\stackrel{\ref{flat:f}}{=}& \G a\lor\F b
  \end{array}
  $$

Intuitively, it means whenever $\paths\models\phi$, it must be the
case that $\paths$ ends up with a loop such that either all states on
the loop satisfy $a$ or at least one state satisfies $b$.
This can be verified by applying the semantics of $\LTL$.

\end{example}
\subsection{Proof of Theorem~\ref{thm:u}}

\begin{proof}
Let $\phi_i$ with $1\le i\le 4$ be any $\LTL$ formula. We have the
following equivalence relations with $*\in\{\land,\lor\}$.
$$
\begin{array}{rcl}
  \F\F\phi&\equiv_c&\F\phi \\
 \F(\phi_1\U\phi_2)&\equiv_c&\F\phi_2\\
  \X(\phi_1 * \phi_2) & \equiv_c & \X\phi_1 * \X\phi_2\\
  \F(\phi_1\lor\phi_2)&\equiv_c&\F\phi_1\lor\F\phi_2 \\
   \F (\phi_1 \land \F \phi_2) & \equiv_c & \F \phi_1 \land \F
     \phi_2\\
  \F (\phi_1 \land \G \phi_2) & \equiv_c & \F\phi_1 \land \G\phi_2\\
   \phi_1 \U (\phi_2\lor\phi_3) & \equiv_c & (\phi_1 \U \phi_2)\lor (\phi_1 \U \phi_3)\\
  (\phi_1 \land\phi_2) \U \phi_3 & \equiv_c & (\phi_1 \U \phi_3)\land
  (\phi_2 \U \phi_3) \\
  \phi_1\U(\phi_2 * \F\phi_3) & \equiv_c &(\phi_1\U\phi_2) * \F\phi_3\\
  \phi_1\U(\phi_2 * \G\phi_3) & \equiv_c & (\phi_1\U\phi_2) * \G\phi_3
\end{array}
$$

$$
\begin{array}{rcl}
  \G\G\phi & \equiv_c &\G\phi\\
 \G(\phi_1\U\phi_2) & \equiv_c
 &\G(\phi_1\lor\phi_2)\land\F\phi_2\\
 \X\F\phi
  \equiv_c \F\phi & \qquad & \X\G\phi\equiv_c\G\phi\\
    \G(\phi_1\land\phi_2)&\equiv_c&\G\phi_1\land\G\phi_2\\
  \G (\phi_1 \lor \G \phi_2)& \equiv_c& \G \phi_1 \lor \G
     \phi_2\\
  \G (\phi_1 \lor \F \phi_2) & \equiv_c & \G\phi_1 \lor \F\phi_2\\
  \qquad(\phi_1 \lor
    \G\phi_2)\U\phi_3 & \equiv_c &(\G\phi_2\land\F\phi_3)\lor(\phi_1\U\phi_3)\\
 (\phi_1 \lor \F\phi_2)\U\phi_3 & \equiv_c & (\F\phi_2\land\F\phi_3)
  \lor (\phi_1\U\phi_3)\\
 (\phi_1 \land \F\phi_2)\U\phi_3 & \equiv_c & (\F\phi_2 \land
  (\phi_1\U\phi_3)) \lor \phi_3 \\
(\phi_1 \land \G\phi_2)\U\phi_3 &
  \equiv_c &(\G\phi_2\land (\phi_1\U\phi_3))\lor\phi_3
\end{array}
$$
Since $\F(\phi_1\lor\phi_2)\equiv\F\phi_1\lor\F\phi_2$ and
$\G(\phi_1\land\phi_2)\equiv\G\phi_1\land\G\phi_2$, together with the
distributive laws of $\land$ and $\lor$, we can
complete the proof. We only show the proofs of the following cases and
omit others which are either similar or trivial.
\begin{itemize}
\item $\F(\phi_1\U\phi_2)\equiv_c\F\phi_2$: $\sigma$ is any infinite word.\\
$\Rightarrow$:
For any $\sigma\models\F(\phi_1\U\phi_2)$, there exists $j\geq 0$ such that $\suffix{\sigma}{j}\models\phi_2$
,which implies $\sigma\models\F\phi_2$.\\
$\Leftarrow$:
For any $\sigma\models\F\phi_2$, there exists $j\geq 0$ such that $\suffix{\sigma}{j}\models\phi_2$, which
implies $\suffix{\sigma}{j}\models\phi_1\U\phi_2$. Therefore, $\sigma\models\F(\phi_1\U\phi_2)$.
\\
Actually we have proved $\F(\phi_1\U\phi_2)\equiv\F\phi_2$ .
\item $\G(\phi_1\U\phi_2) \equiv_c \G(\phi_1\lor\phi_2)\land\F\phi_2$:\\
$\Rightarrow$: $\sigma$ is any infinite word. \\
For any $\sigma\models\G(\phi_1\U\phi_2)$, we have $\forall i\geq 0. \suffix{\sigma}{i}\models\phi_1\U\phi_2$
by definition. Thus, there must exist $j\geq0$ such that $\suffix{\sigma}{j}\models\phi_2$, which implies
$\sigma\models\F\phi_2$. Moreover, suppose $\sigma\nvDash\G(\phi_1\lor\phi_2)$, i.e, $\sigma\models\F(\neg\phi_1\land\neg\phi_2)$, then there exists $k\geq 0$ such that $\suffix{\sigma}{k}\models\neg\phi_1\land\neg\phi_2$, which implies $\suffix{\sigma}{k}\nvDash\phi_1\U\phi_2$. It is
contradictory to $\forall i\geq 0. \suffix{\sigma}{i}\models\phi_1\U\phi_2$, then this direction has been proved.
\\
$\Leftarrow$: $\sigma\in(2^\AP)^{*}$, \\
For any $\sigma^\omega\models\G(\phi_1\U\phi_2)\land\F\phi_2$, we have $\suffix{\sigma^\omega}{i}\models\phi_1\lor\phi_2$ for any $i\geq 0$ and there exists $j\geq 0$ such that
$\suffix{\sigma^\omega}{j}\models\phi_2$. For any $i\geq 0$, we divide it into two cases: 1) if $\suffix{\sigma^\omega}{i}\models\phi_2$, then $\suffix{\sigma^\omega}{i}\models\phi_1\U\phi_2$;
2) if $\suffix{\sigma^\omega}{i}\models\phi_1$, since $\sigma^{\omega}\models\G(\phi_1\lor\phi_2)$ and we want to prove $\suffix{\sigma^\omega}{i}\models\phi_1\U\phi_2$, the remaining question is whether there exists $j\geq i$ such
that $\sigma^\omega\models\phi_2$. Since $\exists j\geq 0. \suffix{\sigma^\omega}{j}\models \phi_2$ and $\sigma^\omega$ is cyclic, it will always be true. Therefore, $\suffix{\sigma^\omega}{i}\models
\phi_1\U\phi_2$.
\item $\F(\phi_1\land\F\phi_2)\equiv_c\F\phi_1\land\F\phi_2$\\
$\Rightarrow$: For any word $\sigma\models\F(\phi_1\land\F\phi_2)$ implies $\exists i\geq 0$ and $j\geq i$ such that $\suffix{\sigma}{i}\models\phi_1$ and $\suffix{\sigma}{j}\models\phi_2$ respectively. Immediately we have $\sigma\models\F\phi_1\land\F\phi_2$.\\
$\Leftarrow$: $\sigma \in(2^\AP)^*$\\ 
$\sigma^\omega\models\F\phi_1\land\F\phi_2$, then $\exists j\geq i \geq 0$ such that $\suffix{\sigma^\omega}{i}\models\phi_1$ and $\suffix{\sigma^\omega}{j}\models\phi_2$ since $\sigma^\omega$ is cyclic word.
Thus $\sigma^\omega\models\F(\phi_1\land\F\phi_2)$.
\item $\F(\phi_1\land\G\phi_2)\equiv_c\F\phi_1\land\G\phi_2$\\
$\Rightarrow$: $\sigma\in(2^\AP)^*$\\ 
$\sigma^\omega\models\F(\phi_1\land\G\phi_2)$, then $\exists i\geq 0$ such that
$\suffix{\sigma^\omega}{i}\models\phi_1\land\G\phi_2$. Clearly, $\sigma^\omega\models\F\phi_1$.
Since $\sigma^\omega$ is cyclic word and $\suffix{\sigma^\omega}{i}\models\G\phi_2$, which implies
$\sigma^\omega\models\G\phi_2$.\\
$\Leftarrow$: $\sigma$ is any infinite word.\\
$\sigma\models\F\phi_1\land\G\phi_2$, then $\exists i\geq 0$ such that $\suffix{\sigma}{i}\models\phi_1$.
Since $\sigma\models\G\phi_2$, we have $\suffix{\sigma}{i}\models\G\phi_2$. Therefore $\suffix{\sigma}{i}\models\phi_1\land\G\phi_2$, which implies $\sigma\models\F(\phi_1\land\G\phi_2)$.
\item $\phi_1\U(\phi_2 \land \F\phi_3)\equiv_c(\phi_1\U\phi_2) \land
  \F\phi_3$. \\
  Let $\phi=\phi_1\U(\phi_2 \land \F\phi_3)$, $\psi=(\phi_1\U\phi_2) \land
  \F\phi_3$ and $\sigma\in(2^\AP)^*$\\
$\Rightarrow$:
  $\sigma^\omega\models\phi$, then
  $\sigma^\omega\models\phi_1\U\phi_2$. Moreover, it must be the case
  that $\exists i\geq 0$ such that
  $\suffix{\sigma^\omega}{i}\models\phi_3$, which implies
  $\sigma^\omega\models\F\phi_3$. Therefore,
  $\sigma^\omega\models\psi$. \\
$\Leftarrow$: $\sigma^\omega\models (\phi_1\U\phi_2)\land\F\phi_3$, then
$\exists j \geq 0$ such that $\suffix{\sigma^\omega}{j}\models\phi_2$ and $\forall 0\leq i <j$,
$\suffix{\sigma^\omega}{i}\models\phi_1$. Since $\sigma^\omega\models\F\phi_3$ and $\sigma^\omega$ is cyclic word, we can always find some $k\geq j$ such that $\suffix{\sigma^\omega}{k}\models\phi_3$, which
implies $\suffix{\sigma^\omega}{j}\models\F\phi_3$. Thus $\sigma^\omega\models\phi_1\U(\phi_2 \land \F\phi_3)$.
\item $\phi_1\U(\phi_2 \lor \F\phi_3)\equiv_c(\phi_1\U\phi_2) \lor
  \F\phi_3$. \\
  Since $\phi_1\U(\phi_2 \lor \F\phi_3)\equiv(\phi_1\U\phi_2)\lor(\phi_1\U(\F\phi_3))$, we only need
  to prove $\phi_1\U(\F\phi_3)\equiv_c\F\phi_3$, which even holds in general.
\item $\phi_1\U(\phi_2 \land \G\phi_3)\equiv_c(\phi_1\U\phi_2) \land
  \G\phi_3$. \\
    Let $\phi=\phi_1\U(\phi_2 \land \G\phi_3)$, $\psi=(\phi_1\U\phi_2) \land
  \G\phi_3$ and $\sigma\in(2^\AP)^*$\\
$\Rightarrow$:
  $\sigma^\omega\models\phi$, then
  $\sigma^\omega\models\phi_1\U\phi_2$. Moreover, $\exists i\geq 0$ such that
  $\suffix{\sigma^\omega}{i}\models\G\phi_3$, which implies
  $\sigma^\omega\models\G\phi_3$. Therefore,
  $\sigma^\omega\models\psi$. \\
$\Leftarrow$: $\sigma^\omega\models (\phi_1\U\phi_2)\land\G\phi_3$, then
$\exists j \geq 0$ such that $\suffix{\sigma^\omega}{j}\models\phi_2$ and $\forall 0\leq i <j$,
$\suffix{\sigma^\omega}{i}\models\phi_1$. Since $\sigma^\omega\models\G\phi_3$ which
implies $\suffix{\sigma^\omega}{j}\models\G\phi_3$. Thus $\sigma^\omega\models\phi_1\U(\phi_2 \land \G\phi_3)$.
\item $\phi_1\U(\phi_2 \lor \G\phi_3)\equiv_c(\phi_1\U\phi_2) \lor
  \G\phi_3$. \\
   Since $\phi_1\U(\phi_2 \lor \G\phi_3)\equiv(\phi_1\U\phi_2)\lor(\phi_1\U(\G\phi_3))$, we only need
  to prove $\phi_1\U(\G\phi_3)\equiv_c\G\phi_3$. Clearly, $\G\phi_3$ implies $\phi_1\U\G\phi_3$.
  Since for any cyclic word $\sigma^\omega\models\phi_1\U\G\phi_3$, $\exists j\geq 0$ such that
  $\suffix{\sigma^\omega}{j}\models\G\phi_3$, which implies $\sigma^\omega\models\G\phi_3$. Thus the claim holds.
\item $(\phi_1 \lor
    \G\phi_2)\U\phi_3\equiv_c(\G\phi_2\land\F\phi_3)\lor(\phi_1\U\phi_3)$. \\
    Let
    $\phi=(\phi_1 \lor \G\phi_2)\U\phi_3$,
    $\psi=(\G\phi_2\land\F\phi_3)\lor(\phi_1\U\phi_3)$ and $\sigma\in(2^\AP)^*$.\\
$\Rightarrow$: $\sigma^\omega\models(\phi_1\lor\G\phi_2)\U\phi_3$, then
    $\exists j\geq 0$ such that $\suffix{\sigma^\omega}{j}\models\phi_3$ and $\forall 0\leq i<j$,
    $\suffix{\sigma}{i}\models\phi_1$ or $\suffix{\sigma^\omega}{i}\models\G\phi_2$. Case i), $\exists 0\leq i<j$
    such that $\suffix{\sigma^\omega}{i}\models\G\phi_2$, which implies $\sigma^\omega\models\G\phi_2$.
    Thus $\sigma^\omega\models\G\phi_2\land\F\phi_3$. Case ii) is trivial, as $\forall 0\geq i<j$,
    $\suffix{\sigma^\omega}{i}\models\phi_2$, which directly conclude $\sigma^\omega\models\phi_1\U\phi_3$.\\
$\Leftarrow$: Since $\sigma^\omega\models\psi$, we
    have either i) $\sigma^\omega\models(\G\phi_2\land\F\phi_3)$ or
    ii) $\sigma^\omega\models(\phi_1\U\phi_3)$. For i), note that
    $\sigma^\omega\models\G\phi_2$ implies all suffixes of
    $\sigma^\omega$ satisfy $\G\phi_2$. Together with the fact that
    $\sigma^\omega\models\F\phi_3$, we have
    $\sigma^\omega\models(\G\phi_2)\U\phi_3$, which implies
    $\sigma^\omega\models\phi$. Case ii) is trivial, as
    $\sigma^\omega\models(\phi_1\U\phi_3)$ implies
    $\sigma^\omega\models\phi$.
\item $(\phi_1 \lor \F\phi_2)\U\phi_3 \equiv_c (\F\phi_2\land\F\phi_3)
  \lor (\phi_1\U\phi_3)$\\
  Let $\phi=(\phi_1 \lor \F\phi_2)\U\phi_3$, $\psi=(\F\phi_2\land\F\phi_3)
  \lor (\phi_1\U\phi_3)$ and $\sigma\in(2^\AP)^*$.\\
 $\Rightarrow$: Since $\sigma^\omega\models\phi$, we have $\exists j\geq 0$ such
 that $\sigma^\omega\models\phi_3$ and $\forall 0 \leq i <j$, $\suffix{\sigma^\omega}{i}\models\phi_1\lor\F\phi_2$. Case i), $\exists 0\leq k<j$ such
 that $\suffix{\sigma^\omega}{k}\models\F\phi_2$, then $\suffix{\sigma^\omega}{k}\models\F\phi_2\land\F\phi_3$.
 Case ii), $\forall 0\leq i<j$, $\suffix{\sigma^\omega}{i}\models\phi_1$, which directly concludes
 $\sigma^\omega\models\phi_1\U\phi_3$. Thus $\sigma^\omega\models\psi$. \\
 $\Leftarrow$: $\sigma^\omega\models\psi$ includes two cases. Case i), $\sigma^\omega\models\F\phi_2\land\F\phi_3$. Since $\sigma^\omega$ is cyclic, $\exists j\geq 0$
 such that $\suffix{\sigma^\omega}{j}\models\phi_2$, then $\forall i\geq 0$, $\suffix{\sigma^\omega}{i}\models\F\phi_2$.
 Thus $\sigma^\omega\models(\F\phi_2)\U\phi_3$, which implies $\sigma^\omega\models(\phi_1\lor\F\phi_2)\U\phi_3$.
 Case ii), $\sigma^\omega\models\phi_1\U\phi_3$, which implies $\sigma^\omega\models(\phi_1\lor\F\phi_2)\U\phi_3$.
\item $(\phi_1 \land \F\phi_2)\U\phi_3  \equiv_c  (\F\phi_2 \land
  (\phi_1\U\phi_3)) \lor \phi_3$\\
Let $\phi=(\phi_1 \land \F\phi_2)\U\phi_3$, $\psi=(\F\phi_2 \land
  (\phi_1\U\phi_3)) \lor \phi_3$ and $\sigma\in(2^\AP)^*$.\\
$\Rightarrow$: Trivial since $\phi$ implies $\phi_3$.\\
$\Leftarrow$: $\sigma^\omega\models\psi$ includes two cases.
Case i), $\sigma^\omega\models\phi_3$, it is obvious that $\psi$ implies $\phi$.
Case ii), $\sigma^\omega\models\F\phi_2\land(\phi_1\U\phi_3)$. Since $\exists j\geq 0$
such that $\sigma^\omega\models\phi_2$, then $\forall i\geq 0$, $\suffix{\sigma^\omega}{i}\models\F\phi_2$.
Thus $\sigma^\omega\models(\phi_1\land\F\phi_2)\U\phi_3$.
\item $(\phi_1 \land \G\phi_2)\U\phi_3  \equiv_c  (\G\phi_2 \land
  (\phi_1\U\phi_3)) \lor \phi_3$\\
 Let $\phi=(\phi_1 \land \F\phi_2)\U\phi_3$, $\psi=(\F\phi_2 \land
  (\phi_1\U\phi_3)) \lor \phi_3$ and $\sigma\in(2^\AP)^*$.\\
$\Rightarrow$: Trivial since $\phi$ implies $\phi_3$.\\
$\Leftarrow$: $\sigma^\omega\models\psi$ consists of two cases.
Case i), $\sigma^\omega\models\phi_3$, it is obvious that $\psi$ implies $\phi$.
Case ii), $\sigma^\omega\models\G\phi_2\land(\phi_1\U\phi_3)$. Since $\sigma^\omega\models\G\phi_2$
implies $\forall i\geq 0$, $\suffix{\sigma^\omega}{i}\models\G\phi_2$.
Thus $\sigma^\omega\models(\phi_1\land\G\phi_2)\U\phi_3$.
\end{itemize}
This completes the proof. Note that using above equations, we can get all $\F$ and $\G$ modalities out
of the scope of the modalities $\U$.
\end{proof}
\subsection{Automaton for $\F\G(a \lor \X(b \U c))$}
We use spot to generate the automaton $\mathcal{A}$ for fair formula $\F\G(a \lor \X(b \U c))$,
which is depicted in Figure~\ref{fig:examplefairLTL} where accepting states are marked by double circles.
Moreover, we add a new atomic proposition
$accepting$ to label accepting states $q_1$ and $q_2$, thus the B\"uchi accepting condition can be represented
as a fair formula $\G\F accepting$.
\begin{figure}[t]
	\centering
	\begin{tikzpicture}[>=stealth',shorten >=1pt,auto]	
    \node (K) at (0,0) {};

    \node[state] (q0) at ($(K) + (0,0)$) {$q_0$};
	\node[state, accepting] (q1) at ($(K) + (2,0.5)$) {$q_1$};
	\node[state, accepting] (q2) at ($(K) + (4, -0.4)$) {$q_2$};
    \node[state] (q3) at ($(K) + (7,1)$) {$q_3$};

	\draw[->] ($(q0.north) + (0,0.3)$) to (q0.north);
	\draw[->] (q0) edge [loop left] node {$\top$} (q0);
	\draw[->] (q0) edge  node {$a$} (q1);
    \draw[->] (q0) edge [bend right]  node[below] {$\neg a$} (q2);
	\draw[->] (q1) edge [loop above ] node {$a$} (q1);
    \draw[->] (q1) edge [bend left ] node {$\neg a$} (q2);
	\draw[->] (q2) edge [bend left] node {$a\land c$} (q1);
    \draw[->] (q2) edge [loop below] node {$\neg a\land c$} (q2);
    \draw[->] (q2) edge  node {$b\land\neg c$} (q3);
    \draw[->] (q3) edge [loop above]  node {$b\land\neg c$} (q3);
    \draw[->] (q3) edge [bend right]  node {$a\land c$} (q1);
    \draw[->] (q3) edge [bend left]  node {$\neg a\land\neg c$} (q2);
	\end{tikzpicture}
	\caption{B\"uchi automaton for $\F\G(a \lor \X(b \U c))$}
	\label{fig:examplefairLTL}
\end{figure}
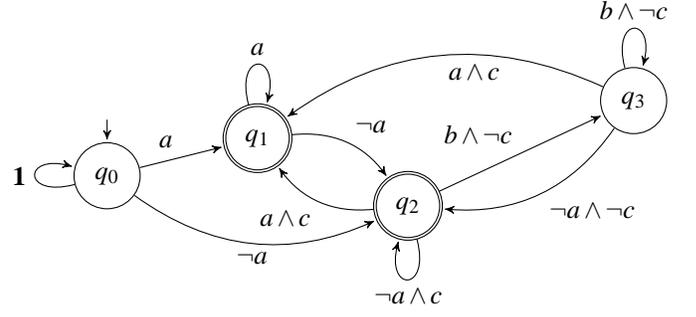

According to the algorithm in the paper, we then compose $\Kripke$ from
Example~\ref{ex:kripke} and $\mathcal{A}$, which gives us the product Kripke
structure $\Kripke'$ in Figure~\ref{fig:examplePro}. After above preparations,
we are trying to find an SCC accepted by $\G\F accepting \land \G\F \neg b$.
There are two alternative SCCs meet the requirement, namely $\{(q_1, s_0)\}$
and $\{(q_1, s_0), (q_1, s_1), (q_2,s_2)\}$. We take SCC $\{(q_1, s_0), (q_1, s_1), (q_2,s_2)\}$ and
construct the counterexample according to the steps we give in the paper.
We can get a counterexample like $(q_0, s_0)((q_1, s_1)(q_2,s_2)(q_1, s_0))^\omega$, which
gives us a corresponding path $s_0(s_1 s_2 s_0)^\omega$ in $\Kripke$.
\begin{figure}[t]
	\centering
    \scriptsize
	\begin{tikzpicture}[>=stealth',shorten >=1pt,auto]	
    \node (K) at (0,0) {};

    \node[state] (q0s0) at ($(K) + (0,3)$) {$(q_0,s_0)$};
    \node[state] (q0s1) at ($(K) + (2,5)$) {$(q_0,s_1)$};
    \node[state] (q0s2) at ($(K) + (5,5.5)$) {$(q_0,s_2)$};
	\node[state, accepting] (q1s1) at ($(K) + (2, 3)$) {$(q_1,s_1)$};
    \node[state, accepting] (q1s0) at ($(K) + (2, 1)$) {$(q_1,s_0)$};
	\node[state, accepting] (q2s2) at ($(K) + (4, 3)$) {$(q_2,s_2)$};

	\draw[->] ($(q0s0.west) + (-0.3,0)$) to (q0s0.west);
    \draw[->] (q0s0) edge [loop below] node {} (q0s0);
    \draw[->] (q0s0) edge node {} (q0s1);
    \draw[->] (q0s0) edge node {} (q1s1);
    \draw[->] (q0s0) edge node {} (q1s0);

    \draw[->] (q0s1) edge node {} (q0s2);
    \draw[->] (q0s1) edge node {} (q2s2);
    \draw[->] (q0s2) edge[bend right=70] node {} (q0s0);
    \draw[->] (q0s2) edge[bend left=50] node {} (q1s0);
    \draw[->] (q1s1) edge node {} (q2s2);
    \draw[->] (q1s0) edge [loop below] node {} (q1s0);
    \draw[->] (q1s0) edge node {} (q1s1);
    \draw[->] (q2s2) edge node {} (q1s0);

	\end{tikzpicture}
	\caption{product Kripke structure $\Kripke'$ for $\Kripke$ and $\mathcal{A}$}
	\label{fig:examplePro}
\end{figure}
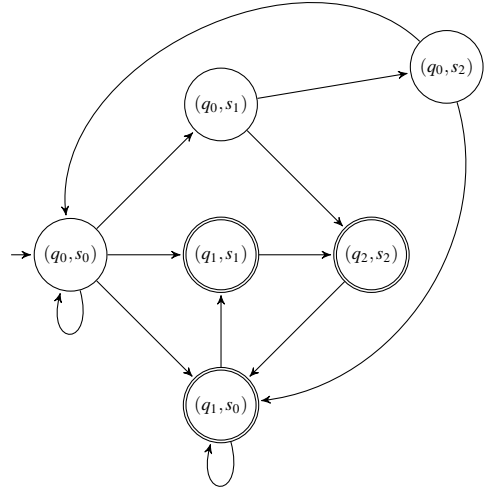
\end{document}